\documentclass[11pt, reqno]{amsart}
\usepackage{latexsym}
\usepackage[english]{babel}
\usepackage{amssymb,amsmath,amsfonts,amscd}
\usepackage{mathrsfs, verbatim, stmaryrd,mathrsfs, enumitem}
\usepackage{graphics,graphicx,color,fancyhdr}
\usepackage{color}
\usepackage{xcolor}
\usepackage{cases}
\usepackage{cite}

\newcommand{\R}{\mathbb{R}}
\newcommand{\dd}{,{\dots},}

\newcommand{\ep}{\varepsilon}
\newcommand{\E}{\mathbb{E}}
\newcommand{\eq}[1]{\begin{equation}
\begin{split}
#1
\end{split}
\end{equation}}
\newcommand{\eqh}[1]{\begin{equation*}
\begin{split}
#1
\end{split}
\end{equation*}}
\newcommand{\lr}[1]{\left( #1 \right)}
\newcommand{\pix}{\rho_0}
\newcommand{\pia}{\pi_{a}}

\newcommand{\lv}{\left\vert}
\newcommand{\rv}{\right\vert}
\renewcommand{\qed}{{ \hfill
                       {\unskip\kern 6pt\penalty 500
                       \raise -2pt\hbox{\vrule\vbox to 6pt{\hrule width 6pt
                       \vfill\hrule}\vrule} \par}  \medskip }}
\newtheorem{theorem}{Theorem}[section]
\newtheorem{lemma}[theorem]{Lemma}

\newtheorem{Note}[theorem]{Note}

\newtheorem{coro}[theorem]{Corollary}

\newcommand{\cP}{\mathcal{P}}
\newcommand{\pa}{\partial}
\newcommand{\ta}{\mathcal{A}^N}
\newcommand{\cL}{\mathcal{L}}

\newcommand{\cLs}{\mathcal{L}_S}
\newcommand{\cLf}{\mathcal{L}_F}
\newcommand{\dr}{\mathcal{F}}
\newcommand{\adr}{\bar{\mathcal{F}}}
\newcommand{\bfa}{{\bf{a}}}
\newcommand{\bfx}{{\bf{x}}}
\newcommand{\mui}{\mu_{\bfx}}

\newcommand{\cPbar}{{\bar\cP^N}}
\newcommand{\Xbar}{\bar{X}^N}
\newcommand{\fone}{f_t^1}
\newcommand{\Kbar}{\bar K^N{}}

\newcommand{\Ione}{\mathcal{I}_1^N}
\newcommand{\Itwo}{\mathcal{I}_2^N}
\begin{document} 
\title{Fast non mean-field networks: uniform in time averaging }

\author{J. Barr\'e$^{(1)}$, P. Dobson$^{(2)}$, M. Ottobre$^{(3)}$ and E. Zatorska$^{(4)}$}
\address{(1) Institut Denis Poisson, Universit\'e d'Orl\'eans, CNRS, Universit\'e de Tours, France and Institut Universitaire de France, julien.barre@univ-orleans.fr}
\address{(2) Delft Institute of Applied Mathematics, Delft University of Technology,
	2628 XE Delft, The Netherlands, p.dobson@tudelft.nl}
\address{(3) Heriot-Watt University, Mathematics Department, Edinburgh EH14 4AS, UK. m.ottobre at hw.ac.uk}
\address{(4) University College London,London WC1E 6BT, UK.  e.zatorska@ucl.ac.uk}
\maketitle
\begin{abstract}
    We study a population of $N$ particles, which evolve according to a diffusion process and interact through a dynamical network. In turn, the evolution of the network  is coupled to the particles' positions. In contrast with the mean-field regime, in which each particle interacts with every other particle, i.e. with $O(N)$ particles, we consider the a priori more difficult case of a {\em sparse network}; that is,  each particle interacts, on average, with $O(1)$ particles.   We also  assume that the network's dynamics is much faster than  the particles' dynamics, with the time-scale of the network described  by a parameter $\ep>0$. We combine the averaging ($\ep \rightarrow 0$) and  the many particles ($N \rightarrow \infty$) limits and  prove that the evolution of the particles' empirical density is described (after taking both limits) by a non-linear Fokker-Planck equation;  we moreover  give conditions under which such limits  can be taken {\em uniformly in time}, hence providing a criterion under which the limiting non-linear Fokker-Planck equation is a good approximation of the original system uniformly in time.  The heart of our proof consists of controlling precisely the dependence in $N$ of the averaging  estimates.
    
    \vspace{5pt}
    
\vspace{5pt}
{\sc AMS Classification (MSC 2020).} 60K35, 47D07, 60J60, 35Q84, 35Q82, 82C31.
\end{abstract}

\section{Introduction}
Systems of diffusing particles with  mean-field interactions, i.e. ``all to all"-type interactions,  have been extensively studied in various contexts. In particular it is well-known that, under mild conditions, the empirical density of  $N$ such particles can be described in the large $N$ limit by a non-linear Fokker-Planck equation;  classical references for this  are \cite{Sznitman, Meleard}. The fundamental reason behind this behavior is a law of large numbers: each particle interacts with many others, hence feels an almost deterministic force.  
A natural question,  which is relevant for modelling purposes, regards interactions mediated by a network where two particles, or more generally two agents, interact only if they are linked together. 
When each particle is linked, on average, to $O(N)$ particles --i.e., borrowing nomenclature from \cite{Dorogovtsev13}, when the {\em average degree}  of the interaction network diverges -- each particle again typically interacts with many others, hence one may still expect a kind of law of large numbers to apply, leading to a non-linear Fokker-Planck equation in the macroscopic limit. With some caveats and subtleties, this is indeed the case, as shown in \cite{Delattre16,Lucon18}; \cite{Coppini19} and \cite{Oliveira19} go further in this direction, and tackle the large deviation regime.
On the other hand, when the mean degree of the interaction network is bounded (in $N$), i.e. when each particle interacts on average with $O(1)$ particles so that the network is {\em sparse},  one may expect a very different, and more complicated, phenomenology; this case is studied e.g. in \cite{Oliveira18}.\\
In many situations, it is actually important to take one further step, and  consider a dynamical network whose evolution is coupled with the particles' dynamics (for reviews see for instance \cite{Gross07,Porter19}). This is clear in epidemiology, where the network of contacts between people is influenced by the spreading of the disease \cite{Gross06,Marceau10}; in models of opinion formation, where people with similar opinions may tend to interact more together \cite{Zanette06}; but also in the modelling of various biological phenomena \cite{Degond16,Taylor17,Peurichard17,Taylor19,BDPZ} (note that \cite{Degond16,Peurichard17} tackle the extra difficulty of considering individual units which are fibers, hence with an orientational degree of freedom, and not merely point particles). 

Details of the model that we will consider  can be found in  Section \ref{subsec:modeldescription} below; for the time-being we summarize the main features of the dynamics that we will be studying:  in this work, which is inspired by the biological applications of  \cite{Degond16,Taylor17,Peurichard17,Taylor19} and by the framework of \cite{BDZ,BCDPZ},  we consider a system of $N$ particles, each of them 
moving in $\R^n$,  and interacting through a sparse graph/network (we will use the words network and graph interchangeably), i.e. particle $i$ and $j$ interact only when they are linked by the edge of an evolving network.  In turn, the evolution of the network is coupled to the particles' dynamics through the following mechanism:  the network's edges are created and deleted at random times described by  Poisson processes; in particular, the rate of the Poisson process governing the link formation mechanism between particle $i$ and particle $j$ depends on the distance between such particles and it is designed in a way that  a link can only be formed if the particles are close enough (i.e., within a fixed distance $R>0$ of each other). Moreover,  our choice of the link formation and destruction rates guarantees that the graph remains sparse.  

These coupled network/particles systems are difficult to analyze mathematically in full generality, and it is then important to understand some limit regimes which may be amenable to analysis. 
In this article, we study the limit of a time-scale separation between network and particles' dynamics and focus on the regime in which the network's dynamics evolves much faster than the particles' dynamics; the time-scale of the network is governed by a parameter $\varepsilon>0$, which enters the evolution in such a way that $\ep\ll 1$ corresponds to accelerating the network's  dynamics.  In this regime one can employ averaging techniques to ``average out" the effect of the network and consider an {\em effective} interaction between particles.  In particular, we take both the averaging ($\ep\rightarrow 0$) and many particles ($N \rightarrow \infty$) limits and prove that the {\em effective macroscopic description} of the dynamics (i.e. the one obtained after taking both the $\ep$ and $N$ limits) is given by an appropriate non-linear Fokker-Planck equation.  Let us emphasize from now that we take such limits {\em together};  this is not only needed if one is after a rigorous treatment but, on a technical level, this is what enables us to treat the a priori more difficult case of a sparse network (see comments in the next paragraph below,  at the end of  Section \ref{sec:main} and at  the beginning of Section \ref{Sec:MF} on this point). Furthermore, since one is often interested in the asymptotic ($t\rightarrow \infty$) behavior of such models, we  identify sufficient conditions so that the { effective macroscopic description}  is a  valid approximation of the initial system  {\em uniformly in time.}  We will comment on the relevance of such uniform in time estimates in Note \ref{Note:mainthm}; for the time being we point out in passing that, as a bi-product, our approach produces the first (to the best of our knowledge) uniform in time averaging result. 

\smallskip
A comprehensive literature review of the subject is out of reach, so we mention only the works which are most relevant to our context. The papers  
\cite{Degond16,BDZ,BCDPZ,Taylor17} consider first the limit of a large population ($N \rightarrow \infty$ with our notation), without the hypothesis of a fast network dynamics, obtaining, in the limit, kinetic-like equations describing the density of the law of the coupled evolution of the particles and the network. To do so, these articles make  formal ``propagation-of-chaos"-like ansatzes, which may be reasonable approximations, but are unlikely to be fully valid in the context we consider here of sparse networks. The works \cite{Degond16,BDZ,BCDPZ} then obtain an effective macroscopic description by taking formally the fast network (i.e. $\ep \rightarrow 0$) limit in these kinetic equations. They obtain the  analog of the nonlinear Fokker-Planck equation  that we obtain in this paper, equation \eqref{macro}. 
In \cite{BDPZ} it was suggested that one should be able to rigorously derive the effective macroscopic description \eqref{macro}, without the formal intermediate step of the kinetic equations (i.e. without any formal ansatz). This is precisely what we achieve in this paper.  

Let us also observe that this interplay between many particle limit and a fast randomly evolving interaction is relevant in the context of two-layers neural networks as well, as new exciting developments are linking the theory of neural networks with the framework of interacting particle systems \cite{Montanari18,VdE18,Spil}: the units of the neural network are described by parameters to be optimized using the available data; the dynamics induced by the learning procedure (a stochastic gradient descent) translates as an interacting particles dynamics for these parameters, in which the interaction is random and changes as new learning data are taken into account. {Further recent relevant literature on mean-field limits for dynamically evolving graphs can be found in \cite{Bhamidi,Bay}.}

\medskip
Our main result is Theorem \ref{thm:main1}, which proves the convergence of the empirical density of particles to the solution of an appropriate  non-linear Fokker-Planck equation (namely equation \eqref{macro} below), in the combined limit $N\to \infty$ (large population) and $\ep \rightarrow 0$ (large time scale separation). Part b of the theorem provides sufficient conditions under which this convergence is uniform in time. Let us briefly comment on the  main technical difficulties which are tackled in this paper.  
Firstly, the proof combines averaging techniques and ``mean-field-like" ones; by themselves, these techniques are known, so the issue here is to combine them and  control precisely the dependence in $N$ of the averaging estimates, in order to be able to take the many particles limit (the key result in this respect is given by Theorem \ref{theo:semigroup}). We believe that the $N$ and $\varepsilon$ dependence of our estimates is optimal, and we heuristically explain at the end of Note \ref{Note:mainthm} why we think this is the case.  Secondly,  the idea  to obtain uniform in time  estimates is borrowed and re-adapted from \cite{CDO,CrisanOttobre} and it hinges on  conquering exponentially fast decay (in time) for the space-derivatives of an appropriate  Markov Semigroup; 
further to \cite{CDO, CrisanOttobre},   we gained explicit dimension-dependence (i.e. dependence on $N$) of the constants appearing in such estimates, which was crucial to our proof.    More comments on this and on related literature in Note \ref{Note:mainthm}.
Finally, to tackle  the many particles limit  in the sparse-network regime, we show that the dynamics ``preserves sparsity": if the network is sparse at time zero, it will remain sparse at subsequent times (this is addressed in Lemma \ref{lem:Ione}). 
To conclude, we point out that, while the set up of the problem is stochastic, the whole method of proof relies solely on analytic techniques. 

\medskip
The paper is organized as follows. In Section \ref{sec:description} we describe the model that we will be analysing, present an intuitive derivation of the  effective  macroscopic description, and state our main result;  Section \ref{sec:main}, Section \ref{Sec:Av} and Section \ref{Sec:MF} are devoted to the proofs: Section \ref{sec:main} outlines the general strategy to prove Theorem \ref{thm:main1}, Section \ref{Sec:Av} details the averaging estimates with the crucial proof of Theorem \ref{theo:semigroup}, and Section \ref{Sec:MF} presents the many-particles estimates.

\section{Description of the model and statement of main results}
\label{sec:description}
\subsection{Description of the Model}\label{subsec:modeldescription}
We consider a system of $N$ interacting particles, each of them  moving in $\R^n$. The position of the $i-$th particle is denoted by $X_t^{i,N, \ep}$ (the reason for the superscript $\ep$ will be clarified below) and, collectively, we will use the notation  $\R^{Nn} \ni X^{N,\ep}_t:=(X^{1,N,\ep}_t, \ldots,X^{N,N,\ep}_t)$ for the whole process; such particles  are connected through an evolving graph which we describe by its adjacency matrix $A^{N,\ep}(t)=\{A_{ij}^{N,\ep}(t)\}_{1\leq i,j\leq N}$. For each $t\geq 0$, $A^{N, \ep}(t)$ is a symmetric $N\times N$ matrix with entries $0$ or $1$: for any pair  $(i,j)$, there is an edge between particles 
$i$ and $j$ at time $t$ if and only if $A_{ij}^{N, \ep}(t)=1$. The particles are not linked to themselves, and so we take $A_{ii}^{N, \ep}(t)=0$ for all $i=1,\ldots,N$. In other words, for every $t \geq 0$, $A^{N, \ep}(t) \in \ta$, where
$$
\ta\!:=\!\{\mbox{arrays } \bfa=\{a_{ij}\}_{1\leq i,j \leq N}\!:\!  a_{ij}=a_{ji} \!\in\!\{0,1\}\  \mbox{and}\ a_{ii}=0 \mbox{ for every } i=1,\ldots,N\}.
$$
The coupled equations of motion for the particles and the evolution of the graph are as follows:
\begin{subequations}\label{mainsf}
	\begin{align}
	dX_t^{i,N,\ep} &= - \nabla V(X_t^{i,N,\ep})dt+ \sum_{j\neq i} A_{ij}^{N,\ep}(t) K(X_t^{j,N,\ep}-X_t^{i,N,\ep}) dt +\sqrt{2D}dB_t^i  \label{slow}\\
	dA_{ij}^{N, \ep}(t) &= -A_{ij}^{N, \ep}(t_-)dN_{ij}^{d, \ep}(t) +[1-A_{ij}^{N, \ep}(t_-)]  dN_{ij}^{f,\ep}(t)\,.  \label{fast}
	\end{align}
\end{subequations}
We describe in turn all the contributions to the above dynamics and make assumptions on the initial data:
\begin{itemize}
	\item all the particles are subject to the same  {\em external potential}  $V:\R^n \rightarrow \R $ (as customary  $\nabla$ denotes $n$-dimensional gradient w.r.t. the  argument of $V$) and to randomness, as described by the processes $B_t^i$'s, $i=1 \dd  N$, which are  $n$-dimensional independent standard Brownian motions; $D>0$ is the diffusion constant;
	\item if  $A_{ij}^{N, \ep}(t)=1$ then at time $t$ particle $i$ and particle $j$  interact through the  {\em interaction force} $K:\R^n \rightarrow \R^n $ and they keep interacting as long as $A_{ij}^{N,\ep}(t)$ remains equal to one;
	\item in \eqref{fast},  $N_{ij}^{d,\ep}, N_{ij}^{f,\ep}$ are independent Poisson processes with intensities $\nu_d \ep^{-1}$ and  $\tilde{\nu}_f(x^i, x^j)\ep^{-1}$, respectively, where: $\ep$ is a fixed positive parameter;  $\nu^d>0$ is a constant (independent of $N, \ep,i$ and $j$), which we refer to as the {\em destruction rate};   $\tilde{\nu}_f: \R^n \times \R^n \rightarrow \R$ is a function, namely
	\begin{equation}\label{nutilde}
	\tilde{\nu}_f(x^i,x^j):= \nu_f^N \varphi_R (|x^i-x^j|)
	\end{equation} 
	with $\nu_f^N$ a positive constant, the {\em formation rate}, which depends on $N$ in a way that will be specified later (the $N$ dependence is chosen in  a way to ensure that the interaction network is sparse, see \eqref{Nscaling}) but is independent of everything else;  $|\cdot|$ denotes euclidean norm  in $\R^n$ and  $\varphi_R: \R \rightarrow [0,1]$ is a smooth  function with compact support in $[0,R]$ for some fixed $R>0$;
	\item the initial data $X_0^{N, \ep}, A^{N, \ep}(0)$ are i.i.d.~random variables;  we will make precise assumptions on the initial data in Section \ref{sec:2}, see Hypothesis \ref{H3new};
	\item finally, $t_-$ denotes left limit at $t$;  that is, for any c\'adl\'ag function $f:\R \rightarrow \R$, we set $f(t_-):= \lim_{s \rightarrow t_-} f(s)$. 
\end{itemize}
Precise assumptions on the potential $V$ and force kernel $K$ and on the initial data will be given in Section \ref{sec:2}. The evolution for the link $A_{ij}^{N,\ep}$ in  \eqref{fast} is a piecewise constant process which switches between the values 0 and 1 at random times according to  the Poisson clocks $N_{ij}^{d,\ep}$ and $N_{ij}^{f,\ep}$. In particular, equation \eqref{fast} is coupled to the particles' evolution \eqref{slow} only through the process $N_{ij}^{f,\ep}(t)$; this is a bit hidden in the notation that we have used in \eqref{fast}, so to clarify this fact let us point out that if $\varphi_R$ was the indicator function of the interval $[0,R]$ (denoted by $\mathbf{1}_{[0,R]}(\cdot)$), \eqref{fast} could be replaced by the following equivalent dynamics
\begin{equation}\label{fastprime}
dA_{ij}^{N, \ep}(t) = -A_{ij}^{N, \ep}(t_-)dN_{ij}^{d, \ep}(t) +[1-A_{ij}^{N, \ep}(t_-)] \mathbf{1}_{[0,R]}\left(\lv X_t^{i,N,\ep}- X_t^{j,N,\ep}\rv\right) d\breve{N}_{ij}^{f,\ep}(t)\, ,
\end{equation}
where now $\breve{N}_{ij}^{f,\ep}(t)$ is a simple Poisson process with rate $\nu_f^N$. 

In order to gain further intuition about the dynamics \eqref{slow}-\eqref{fast} one may think for example of $K$ as being an elastic force (i.e. a short-range repulsion and long-range attraction kind of interaction) so that, if $A_{i,j}^{N, \ep}(t)=1$ then at time $t$ particle $i$ and $j$ are linked via a spring (as well as being subject to the external potential $V$ and to randomness) and  
{when the Poisson clock $N_{ij}^{d,\ep}$ goes off then the link between the two particles is destroyed, i.e.  $A_{ij}^{N, \ep}$ switches to zero.  However, if  particle $i$ and $j$ are not linked at time $t$, i.e.  if $A_{ij}^{N, \ep}(t)=0$, then two things can happen: if 
	$|X^{i,N,\ep}_t - X^{j,N,\ep}_t| \leq R$, then the Poisson clock
	$N_{ij}^{f,\ep}$ may have a positive rate (depending on whether $\varphi_R\left(|X^{i,N,\ep}_t - X^{j,N,\ep}_t|\right)$ is bigger or equal to zero) and may go off, in which case a link is formed and $A_{ij}^{N,\ep}$ switches to one; if on the other hand $|X^{i,N,\ep}_t - X^{j,N,\ep}_t| > R$, then the Poisson clock $N_{ij}^{f,\ep}$ has zero rate (as in this case $\varphi_R\left(|X^{i,N,\ep}_t - X^{j,N,\ep}_t|\right)=0$) and does not go off, so that $A_{ij}^{N,\ep}$ remains equal to  zero at least as long as $|X^{i,N,\ep}_t - X^{j,N,\ep}_t|>R$. In other words the function $\varphi_R$ ensures that an edge/link can only be created between particles $i$ and $j$ if such particles are within distance $R$ of each other.}  

The existence of both weak and strong  solutions to system 
\eqref{slow}-\eqref{fast} is standard under our assumptions on the potentials $V$ and force kernel $K$ (see Note \ref{note:hypothesis}). For completeness, we quickly recall the strategy to build such solutions in Appendix \ref{app:A}.

\subsection{Main Result.} 
\label{sec:result}
We are concerned with understanding the behaviour of  system \eqref{slow}-\eqref{fast} in the limit $\ep\rightarrow 0$ and $N\rightarrow \infty$. Let us emphasize that these limits are taken {\em together} --  although, for expository purposes, in what follows  we present the heuristic derivation of our main results by formally taking the $\ep$ limit first (averaging limit)  and then  the $N\rightarrow \infty$ (large particle) limit. Exact proofs are contained in Section \ref{sec:main}, Section \ref{Sec:Av} and Section \ref{Sec:MF}. 

When $0<\ep \ll
1$, the links' dynamics is much faster than the particles' dynamics; this kind of situation is dealt with within the framework of stochastic averaging, see e.g. \cite{PS}. The  idea of stochastic averaging can be informally described as follows: because the process $A^{N, \ep}(t)$ \eqref{fast}  evolves much faster than the process $X_t^{N, \ep}$ \eqref{slow}, the former process will have effectively reached equilibrium while the latter has remained substantially unchanged.  In particular,  considering the slow process as ``frozen", say $X_t^{N, \ep}=\bfx \in \R^{Nn}$, the fast process will
converge, as $\ep\rightarrow 0$, to a stationary distribution, denoted by $\mui$ (this measure clearly depends on the ``frozen" values of the slow process). 
Note that, given the value of the slow process $X_t^{N, \ep}$, the components $A_{ij}^{N, \ep}(t)$ of the fast process evolve independently of each other; the evolution of each component is a simple two-state Markov process with generator given by the matrix  
$$
L^{(ij)}=\frac{1}{\varepsilon} 
\left( 
\begin{array}{cc}
-\nu_f^N\varphi_R(|x^i-x^j|) & \nu_f^N\varphi_R(|x^i-x^j|) \\
\nu_d & -\nu_d
\end{array}
\right) \,. 
$$
The invariant measure  $\mui$ of  $A^{N, \ep}(t)$ given $X_t^{N,\ep}$, is a probability measure  on the space  $\ta$.  Let us emphasize that this is the invariant measure for \eqref{fast}, given $X_t^{N, \ep}=\bfx$. Because of conditional independence, one just needs to find the invariant measure $\mui^{(ij)}$ for the dynamics of $A_{ij}^{N, \ep}$ (given $X_t^{N, \ep}$) and then the overall invariant measure on $\ta$ is just a product measure. The invariant measure for the process $A_{ij}^{N, \ep}(t)$ is the only normalised element of the kernel of the transpose of the matrix $L^{(ij)}$, see \cite{Norris}. A simple calculation yields$$
\mu^{(ij)}_{\bfx}(1)= \frac{\nu_f^N \,   \varphi_R (\lv x^i-x^j\rv )}{\nu_f^N \varphi_R (\lv x^i-x^j\rv )+\nu_d}, \quad \mu^{(ij)}_\bfx(0)= 1-\frac{\nu_f^N \varphi_R (\lv x^i-x^j\rv )}{\nu_f^N \varphi_R (\lv x^i-x^j\rv )+\nu_d} \,.
$$
Hence, the sought after invariant measure $\mui$ on $\ta$ is given by 
\begin{equation}
\mu_\bfx (\bfa) = \prod_{i\neq j} \left( \frac{\nu_f^N\varphi_R (\lv x^i-x^j\rv ) }{\nu_f^N\varphi_R (\lv x^i-x^j\rv ) +\nu_d} \delta_{a_{ij}-1} + \left[1- \frac{\nu_f^N\varphi_R (\lv x^i-x^j\rv ) }{\nu_f^N\varphi_R (\lv x^i-x^j\rv ) +\nu_d}  \right]\delta_{a_{ij}}\right) \,,
\end{equation}
for $\bfa=\{a_{ij}\} \in \ta$ (as customary,  
$\delta_a$ is the Kronecker delta, i.e $\delta_a=1$ if $a=0$ and it is equal to zero otherwise).

According to the classical averaging paradigm, as $\ep \rightarrow 0$ the fast-slow system \eqref{slow}-\eqref{fast} converges to the so-called  {\em averaged dynamics} (also often referred to as {\em effective dynamics}); the evolution equation for the averaged dynamics $\{\bar X^{i,N}_t\}_{1\leq i\leq N}$ is formally obtained by taking the average of drift and diffusion coefficients of  \eqref{slow}  with respect to the measure $\mui$. Because,  given $X_t^{N, \ep}$, we have 
\begin{align}
&\mathbb{E}_{\mu_{X_t^{N, \ep}}}\left[ \sum_{j\neq i} A_{ij}^{N,\ep}(t) K(X_t^{j,N,\ep}-X_t^{i,N,\ep}) \right] =  \sum_{j\neq i} \mathbb{E}_{\mu_{X_t^{N, \ep}}}\left[A_{ij}^{N, \ep}\right] K(X_t^{j,N, \ep}-X_t^{i,N,\ep}) \\
& = \sum_{j\neq i} \frac{\nu_f^N \varphi_R (|X^{i,N,\ep}-X^{j,N,\ep}|)}{\nu_f^N \varphi_R (|X^{i,N,\ep}-X^{j,N, \ep}|)+\nu_d}  K(X_t^{j,N,\ep}-X_t^{i,N,\ep}),
\end{align}
the resulting equation for $\bar X^{i,N}_t$ is 
\begin{align}\label{slow2}
d\bar X_t^{i,N}\! &=\! - \nabla V\!(\bar X_t^{i,N})dt\!+\!\sum_{j\neq i}\! \frac{\nu_f^N \varphi_R \lr{\lv \bar X_t^{i,N}- \bar X_t^{j,N}\rv }}{\nu_f^N \varphi_R\! \lr{\lv \bar X_t^{i,N}\!-\! \bar X_t^{j,N}\rv }\!+\!\nu_d}   K(\bar X_t^{j,N}\!-\! \bar X_t^{i,N}) dt \!+\!\sqrt{2D} \, dB_t^i  \\
&= - \nabla V(\bar X_t^{i,N})dt+\sum_{j\neq i} \bar K^N(\bar X_t^{j,N}- \bar X_t^{i,N}) dt + \sqrt{2D} \, dB_t^i \, ,
\end{align}
with initial conditions $\bar X_0^{i,N}= X_0^{i,N, \ep}$ for $i \in \{1 \dd N\}$ and 
having set 
\begin{equation}\label{eq:Kbardef}
\bar K^N (x^{j}- x^{i}):= \frac{\nu_f^N \varphi_R (\lv x^i-x^j\rv )}{\nu_f^N \varphi_R (\lv x^i-x^j\rv )+\nu_d}   K(x^j- x^i). 
\end{equation}
For brevity we will later on use a more compact notation for the whole drift of the averaged equation, namely
\begin{equation}\label{fbari}
\adr^i(\bfx):= - \nabla V(x^i) + \sum_{j \neq i} \bar K^N (x^j-x^i) \,.
\end{equation}
We now assume the following  scaling of $\nu_f^N$ with $N$: 
\begin{equation}\label{Nscaling}
\nu_f^N=\nu_f/N\, ,
\end{equation}
where $\nu_f>0$ is a constant (independent of $N, \ep$ and $i,j$). 
{This choice of scaling for $\nu_f$ ensures that the average number of links per particle is at most of order $N\frac{\nu_f^N}{\nu_d+\nu_f^N} =O(1)$. That is, heuristically,   each particle interacts with $O(1)$ particles (on average), so that  the interaction graph is sparse.}
Equation \eqref{slow2} has now the canonical form of interacting diffusions considered e.g. in \cite{Sznitman}. As $N\rightarrow \infty$ the particles become independent of each other and, in the limit,  each of them evolves according to the following $n-$dimensional non-linear diffusion (here the non-linearity is {\em in the sense of McKean})
\begin{equation}\label{interX}
dX_t = - \nabla V(X_t)dt+ \left(\bar{K}\ast \rho_t \right) (X_t) dt +\sqrt{2D}dB_t.
\end{equation}
In the above $\ast$ denotes convolution (over $\R^n$), the function $\bar K:\R^n \rightarrow \R^n$ is defined as
\begin{equation} \label{barK}
\bar K(x)= \frac{\nu_f}{\nu_d} \varphi_R(\lv x \rv)  \, K(x), \quad x \in \R^n \, ,
\end{equation}
and $\rho_t: \R^n \rightarrow \R$ is the (density of the) law of the process $X_t$ itself; such a function  solves the  following non-linear Fokker-Planck equation: 
\eq{ \label{macro}
	\partial_t \rho_t(x) &= \nabla\cdot(\rho_t(x)\nabla V(x))  -\nabla \cdot\left(\rho_t(x) ({\bar K}*  \rho_t)(x)\right)+D \Delta \rho_t(x),\\
	\rho_t(x)|_{t=0}&=\rho_0(x),
}

Note that the expression of the kernel $\bar K$ can be formally obtained by taking sums over $j$ in \eqref{eq:Kbardef}, using \eqref{Nscaling} and letting $N \rightarrow \infty$.  
{Equations such as \eqref{macro} find applications in many areas of natural sciences, and have a rich phenomenology (see \cite{frank} for a textbook reference). It is worth mentioning that if $\bar{K}$ is the gradient of some function, then \eqref{macro} has a gradient flow structure, a very helpful fact to elucidate its qualitative behavior.}

Our main result can then be stated as follows -- the hypothesis under which such a result holds are detailed in Section \ref{sec:2}. 
\begin{theorem}\label{thm:main1}
	With the notation introduced so far, let the potential $V$ and force kernel $K$ satisfy  Hypothesis \ref{H1}.  Let  $\hat{\mu}_{X^{N,\varepsilon}_t}$ be the empirical measure of the particle system $X_t^{i,N, \varepsilon}$ evolving according to \eqref{slow}-\eqref{fast}, namely
	$$\hat{\mu}_{X^{N,\varepsilon}_t}(x):=\frac{1}{N}\sum_{i=1}^N \delta_{X_t^{i, N, \ep}}(x) \, , $$ 
	and let $\rho_t$ be the solution to \eqref{macro}. 
	If  the  initial data for the particle system \eqref{slow}-\eqref{fast} and for the PDE \eqref{macro} satisfy Hypothesis \ref{H3} and Assumption \ref{H3new} then, 
	\begin{description}
		\item[a)]{ for every $u \in C_b^2(\R^n)$ (i.e. twice differentiable, bounded and with bounded first and second derivatives), there exists a constant $C>0$ ( independent of $\varepsilon, T$ and $N$ but dependent on $u$) such that 
			\eq{\label{main1_statement}
				\sup_{t \in [0,T]}\E \left( \int_{\R^n} u(y) \,  d \hat{\mu}_{X^{N,\varepsilon}_t}(y)\, - \, \int_{\R^n} u(y) \rho_t(y) dy  \right)^2
				\leq   \frac{C}{N}e^{CT} + C (1+ T) N \varepsilon ,
			}
			for every $T>0, \, \varepsilon>0$ and $N \in \mathbb N_+$. Hence, as $\ep\rightarrow 0$ and $N\rightarrow \infty$ in such a way that $N\ep \rightarrow 0$, the particle system \eqref{slow}-\eqref{fast} converges to the solution of \eqref{macro}.   
			\item[b)] If, in addition, the external potential  $V$ satisfies Hypothesis \ref{H2} with ${\kappa}_3 \geq \max(\kappa_{\rm av},\kappa_{\rm mf})$ (where $\kappa_{\rm av}$ and $\kappa_{\rm mf}$ are defined in Lemma \ref{lem:derivativeest}, and Lemma \ref{lem:step2}, respectively, and $\kappa_3$ appears in Hypothesis \ref{H2}) then for every $u \in C_b^2(\R^n)$ (there exists a constant $C>0$ (independent of $\varepsilon$, time and $N$ but dependent on $u$) such that 
			\eq{\label{main0_statement}
				\sup_{t \geq0}\E \left( \int_{\R^n} u(y) \,  d \hat{\mu}_{X^{N,\varepsilon}_t}(y)\, - \, \int_{\R^n} u(y) \rho_t(y) dy  \right)^2 
				\leq \frac{C} {N}+CN\ep\, , 
			} 
			for every $\varepsilon>0$ and $N \in \mathbb N_+$.}
		Hence, as $\ep\rightarrow 0$ and $N\rightarrow \infty$ in such a way that $N\ep \rightarrow 0$, the particle system \eqref{slow}-\eqref{fast} converges to the solution of \eqref{macro} uniformly in time.
	\end{description}
\end{theorem}

\begin{Note}\label{Note:mainthm}\textup{Some comments on the above statement. 
		\begin{itemize}
			\item The first part of the theorem proves that the macroscopic density $\rho_t$, solution of \eqref{macro}, is a good approximation of the empirical density of the particles as soon as $N$ is large and $N\varepsilon$ is small, but only on a finite time horizon. Strengthening this result to the {\it uniform in time} estimate \eqref{main0_statement} requires the strong confining properties  stated in Hypothesis \ref{H2}. 
			To prove Theorem \ref{thm:main1}, we need to combine averaging ($\varepsilon \to 0$) and many-particles ($N \to \infty$) estimates. In both \eqref{main1_statement} and \eqref{main0_statement}, the first error term on the RHS comes from the many-particles estimate, and the second from the averaging procedure. In the uniform in time estimate \eqref{main0_statement},
			the requirements $\kappa_3 \geq \kappa_{\rm av}$ and $\kappa_3 \geq \kappa_{\rm mf}$ come respectively from the averaging and many-particles steps.
			One of the main difficulties in the proof is to track precisely the $N$-dependence of the averaging estimates. As already emphasized in the introduction, the $\ep$ and $N$ limits are taken together. We further explain why at the end of Section \ref{sec:main}.
			\item To further elaborate on the significance of the uniform in time result, let us point out that 
			a long-standing criticism of averaging techniques is the following: while one can typically only prove that the averaged dynamics  is a good approximation of the original slow-fast system for {\em finite-time} windows, the averaged dynamics is often used in practice to make predictions about the {\em long-time} behaviour of the slow-fast system. The fact that the averaged dynamics is a good approximation of the slow-fast system also for long times is typically only conjectured (in specific cases) on the basis of numerical evidence. Part b of Theorem \ref{thm:main1} allows to avoid this pitfall.
			\item As already mentioned in the introduction, from a technical point of view the idea used to obtain uniform in time  estimates is borrowed from \cite{CDO} and it requires proving exponentially fast decay (in time) for the space-derivatives of the Markov Semigroup associated to \eqref{slow2}, i.e. estimates of the type \eqref{eq:expdecayofderivative}.   Such estimates, which are related to those   studied in \cite{CrisanOttobre, CCDO1, Dragoni}, are different from the ones usually appearing in the literature in two respects: first, similarly to  \cite{CrisanOttobre, CCDO1}, these are  not smoothing-type  estimates; second, and further to \cite{CrisanOttobre, CCDO1, Dragoni}, the   dimension-dependence (i.e. dependence in $N$)  of the constants appearing in them is explicit.    To further explain the former point, note that in contrast to \eqref{eq:expdecayofderivative}, smoothing estimates for a given Markov Semigroup \footnote{The precise definition of Markov Semigroup is given in Section \ref{sec:2} below.} $\cP_t$ are, generally, of the form
			$$
			\lv \mathfrak{D}\cP_t f(x) \rv \leq \frac{1}{t^{\gamma}} \|f\|_{\infty}, \quad f \in C_b,
			$$
			where $\mathfrak{D}$ is some appropriate differential operator and $\gamma>0$ depends on $\mathfrak{D}$ (for example, if the semigroup is elliptic one can take $\mathfrak{D}$ to be the usual gradient and $\gamma=1/2$, if the process is hypoelliptic $\gamma$ will depend on the number of commutators needed to obtain the direction $\mathfrak{D}$), see \cite{Bakry} and references therein for a comprehensive review. Note that in the above $f \in C_b$ while in \eqref{eq:expdecayofderivative} we take $f \in C_b^2$; smoothing estimates can be seen as quantifying the ``explosion" of heat-type semigroups as $t\rightarrow 0$ so, in a way, they are more meaningful for $t$ small. Here we want a specific (i.e. exponential) quantitative estimate for $t$ large. Note that by the semigroup property the short and long-time estimates could be ``glued" together, and this is routinely done in the literature; we do not  do it here as the smoothing effect is not the main concern of the paper, what we are interested in is the long-time regime. Moreover,  this would imply having to trace the $N$ dependence in the short-time estimate as well, which is not central to the scope of this paper.  
			\item As a bi-product of the above observation, because in \eqref{main1_statement} and \eqref{main0_statement} we take  $f \in C_b^2$ rather than $f \in C_b$, from our result we cannot deduce weak convergence  in $C([0,T]; \mathbb Pr)$, where $\mathbb Pr$ is the space of probability measures; hence,  in the statement of the theorem, by {\em convergence} we mean convergence in the sense of \eqref{main1_statement} (and \eqref{main0_statement}). 
			\item The theorem is meaningful in the regime $N$ large, $\ep$ small and $N\ep$ small. This may be expected from \eqref{Nscaling}: the rate of the Poisson process creating the links is of order $(N\ep)^{-1}$; good averaging estimates then require $N\ep \ll 1$.  
		\end{itemize}
	}
\end{Note}

\subsection{Setting, Notation and Assumptions}\label{sec:2}

{Because our process \eqref{slow}-\eqref{fast} evolves in {$\R^{Nn} \times \ta$}, we will deal with vectors $\bfx \in \R^{Nn}$ -- the boldfont letter is reserved for elements of $\R^{Nn}$. In particular, $\bfx=(x^1\dd x^N)$, where $x^j \in \R^n$ for every $j \in 1\dd N$. Elements  $\bfa \in \ta$ will be regarded as arrays  with components $\bfa=(a_{12}, a_{13}, \dots, a_{NN})$.  As customary, $C_b^k(\R^d)$ will denote the set of real-valued  functions which are $k$-times continuously differentiable, bounded and  with derivatives up to order $k$ all bounded. Regarding initial data, we make the following standing assumption.  
	\begin{enumerate}[label=\textbf{\textup{[A.\arabic*]}}]
		\item \label{H3new}  We  assume  that the initial datum of \eqref{macro}, i.e. $\pix (x)$, is a smooth probability density  on $\R^n$ with finite moments of order two. The initial data for \eqref{slow}-\eqref{fast} are as follows:     $\{X_0^{i,N, \ep}\}_{1\leq i \leq N}$,  are i.i.d. random variables, $X_0^{i,N, \ep} \sim \pix$ for every $1\leq i \leq N$;   the initial distribution of the links is i.i.d.~as well,  $A_{i,j}^{N, \ep}(0)\sim \pia$ for every $1\leq i<j \leq N$, where $\pia$ is any probability distribution on $\{0,1\}$. The initial distribution of the particles,  $\pix$, is independent of the initial distribution of the links, $\pia$, and it is also independent of every other source of 
		randomness. 
	\end{enumerate}
	
	The dynamics \eqref{slow}-\eqref{fast} contains four sources of randomness: the Brownian motions $B(t)=(B^1(t), \dots, B^N(t))$, the Poisson processes driving formation and destruction of links, the initial distribution of the particles $\pix$,  and the initial distribution of the links $\pia$. In the proofs below, we will first consider a fixed initial datum $X_0^{N, \ep}=\bfx, A^{N, \ep}(0)=\bfa$ and then, at the latest possible moment, we  take expectation with respect to the noise in the initial data,  i.e. with respect to $\bfx=(x^1, \dd x^N)$ and $\bfa=(a_{11}, a_{12} \dd a_{NN})$ with  $\bfx \sim \Pi_{i=1}^N \pix$ and $\bfa \sim \Pi_{i,j=1}^N \pia$.    Moreover, 
	\begin{itemize}
		\item $\E_{\pi}$ will denote expectation with respect to both $\bfx\sim \Pi_{i=1}^N \pix$ and $\bfa \sim \Pi_{i,j=1}^N \pia$;
		\item $\E_{B}$ will denote expectation with respect to all the sources of noise except those in the initial data (of both the particles and the links); 
		\item $\E$ is expectation with respect to all the sources of noise.
	\end{itemize}
}
With this notation in mind, let $\cPbar_t: C_b^2(\R^{Nn}) \rightarrow C_b^2(\R^{Nn})$ {\footnote{The semigroup is usually defined on the set of continuous and bounded functions. As mentioned in Note \ref{Note:mainthm}, here we are not interested in smoothing results so, to avoid further technical complications, we just work with $C_b^2$ functions.}} be the  semigroup associated with the SDE \eqref{slow2}, namely
$$
(\cPbar_t f)(\bfx):=\mathbb{E}_B[f(\Xbar_t)|\, \Xbar_0=\bfx]\,;
$$
where there is no risk of confusion,  we will use interchangeably the notation $\cPbar_t f(\bfx)=(\cPbar_t f)(\bfx)$. 
It is well-known that such a semigroup is a classical solution to the PDE
\begin{equation}\label{eq:barPDE}
\begin{cases}\partial_t(\cPbar_t f)(\bfx)=\bar{\cL}^N\,  \cPbar_t f(\bfx)\\
(\cPbar_0 f)(\bfx)=f(\bfx)
\end{cases}
\end{equation}
where $\bar \cL^N$ is the generator of the diffusion $\bar X^N_t$ , i.e. the second order differential operator formally acting on smooth functions as 
\begin{equation}\label{10star}
(\bar\cL^N f)(\bfx) := \sum_{i=1}^N\left[\bar{\mathcal{F}}^i(\bfx) \cdot \nabla_i f(\bfx)\right]+ D \,  {\mathrm Tr} (\mathrm{Hess} f)(\bfx)\,,
\end{equation}
where ${\mathrm Tr} (\mathrm{Hess} f)$ is simply the Trace of the $Nn\times Nn$ Hessian matrix of the function $f$ and $\nabla_i$ is the $n$-dimensional gradient with respect to the component $x^i$ of $\bfx$. 
Analogously, the process $(X_t^{N, \ep}, A^{N, \ep}(t))$ solution of  \eqref{slow}-\eqref{fast} generates a semigroup $\cP_t^{N, \ep}: C_b^2(\R^{nN}\times \ta) \rightarrow C_b^2(\R^{nN}\times \ta)$, namely
$$
(\cP_t^{N,\ep}h)(\bfx, \bfa):= \mathbb E_B\left[ h(X_t^{N, \ep}, A^{N, \ep}(t))\vert (X_0^{N, \ep}, A^{N, \ep}(0))=(\bfx, \bfa) \right]\,.
$$
Because we want to compare the dynamics $\Xbar_t$ with the dynamics $X_t^{N, \ep}$ (but $X_t^{N, \ep}$ alone does not generate a semigroup), for technical convenience we will restrict our attention to the case in which the  functions $h$ on which $\cP_t^{N, \ep}$ acts are constant in the variable $\bfa$, i.e. to the case in which $h$ does not depend on $\bfa$ and for such functions, coherently with the above, we use the letter $f$. Note that while $f$ acts only on the variable $\bfx$, the function $(\cP_t^{N, \ep} f)(\bfx,\bfa)$ depends on both variables. In other words, we will restrict to considering initial value problems where the initial profile is a function independent of $\bfa$: 
\begin{equation}\label{eq:epsPDE}
\begin{cases}\partial_t(\cP_t^{N, \varepsilon} f)(\bfx,\bfa)=\cL^{N, \ep} (\cP_t^{N, \varepsilon} f)(\bfx,\bfa)\\
(\cP_0^{N,\varepsilon} f){(\bfx,\bfa)}=f(\bfx) \,.
\end{cases}
\end{equation}
The generator $\cL^{N, \ep}$ is the operator (formally) defined on sufficiently smooth functions $h:\R^{nN} \times \ta \rightarrow \R$ as 
\eq{\label{gen:Lep}
	\cL^{N,\varepsilon} h:= \cL_S h +\frac{1}{\varepsilon} \cL_F h,
}
where $\cLs$ and $\cLf$ are the parts of the generator corresponding to the slow and to the fast dynamics, respectively (both $\cLs$ and $\cLf$ should be denoted by $\cL_S^{N, \ep}$ and $\cL_F^{N, \ep}$, we don't do this to avoid cumbersome formulas later on); namely, 
\eq{\label{def:Ls}
	(\cLs h)(\bfx,\bfa)\!&:=\!\sum_{i=1}^N\!\left[\!\lr{\!-\nabla V(x^i)+\sum_{j\neq i}
		a_{ij} K(x^j\!-\!x^i)\!} \! \cdot\! \nabla_{i}h(\bfx,\bfa)\!+\!D \, \mathrm{Tr}(\mathrm{Hess}_{\bfx}\,h(\bfx,\bfa))\right]\\
	&:=\sum_{i=1}^N \left[\dr^i(\bfx,\bfa)\cdot \nabla_{i}h(\bfx,\bfa)+ D
	\mathrm{Tr}(\mathrm{Hess}_{\bfx}\, h(\bfx,\bfa))\right]
}
having set
$$
\dr^i(\bfx,\bfa):= -\nabla V(x^i)+\sum_{j\neq i}
a_{ij} K(x^j-x^i) \, ,
$$
and
\begin{align}\label{def:Lf}
(\cLf h)(\bfx,\bfa)= h_d(\bfx,\bfa)   +  \, h_f(\bfx,\bfa) \, , 
\end{align}
where
\begin{align}
h_d(\bfx,\bfa) &:= \nu_d \sum_{(ij):i<j} h_d^{(ij)}(\bfx,\bfa), \\
\qquad h_f(\bfx,\bfa) &:= \nu_f^N \sum_{(ij):i<j} \varphi_R (|x^i-x^j|)h_f^{(ij)}(\bfx,\bfa), \\
h_d^{(ij)}(\bfx,\bfa)&:=\left[ h(\bfx, a_{12}, \dots,\underbrace{ 0}_{ij}, \dots, a_{NN}) - h(\bfx,\bfa) \right], \label{l1}\\ 
h_f^{(ij)}(x,\bfa) &:=  
\left[ h(\bfx, a_{12}, \dots, \underbrace{1}_{ij}, \dots, a_{NN}) - h(\bfx,\bfa) \right].\label{l2}
\end{align}
Note that the operator $\cLs$ acts on the variable $\bfx$ only, while the operator $\cLf$ acts on the variable $\bfa$ only.
In each addend of line \eqref{l1} (line \eqref{l2}, respectively), the component $a_{ij}$ of $\bfa$ has been replaced by 0 (1, respectively) and  we have used interchangeably the notation $h(\bfx,\bfa)$ and $h(\bfx,a_{12}, \dots, a_{NN})$. It remains understood that $a_{ij}=a_{ji}$ so in \eqref{l1} and \eqref{l2}, we effectively set to zero (or to one)  the value of two components.

Let us now move on to stating the following assumptions for the kernel $K$, the confinement potential $V$ { and the initial data}.   
\begin{enumerate}[label=\textbf{\textup{[H.\arabic*]}}]
	\item \label{H1} The interaction force
	$K : \R^n \rightarrow \R^n$ and the external potential $V : \R^n \rightarrow \R$ are such that
	\begin{enumerate}[label=\textbf{\textup{[H1.\arabic*]}}]
		\item \label{H1K} $K \in C^2(\R^n)$ and its first and second derivatives are bounded.
		\item \label{H1V} $V \in  C^3(\R^n)$ and its second and third derivatives are bounded.
	\end{enumerate}
	
	\item \label{H2}
	$V$ is strongly convex:
	\[
	{\rm Hess} V(x) \geq \kappa_3 I_n,
	\]
	where $\kappa_3>0$ will be chosen to be  sufficiently large (see  Theorem \ref{thm:main1}), $\rm{Hess} V$ is the Hessian matrix of $V$, and the inequality is in the sense of quadratic forms.
	\item \label{H3} The initial data $X_0^{N, \ep} =\bfx \in \R^{nN}, A^N(0)=\bfa \in \ta$ of system \eqref{slow}-\eqref{fast} are as follows:
	\begin{itemize}
		\item If the interaction force $K$ satisfies \ref{H1} and is unbounded then we assume that  there exists  a function $g\in C^2(\R^n)$, $g(x) \!\geq \max\{1,\lvert K(x)\rvert, \lvert x\rvert\}$, and  with first and second bounded derivatives such that, after setting  
		\eq{\label{def:psii}
			\psi^i(\bfx,\bfa) &:=\sum_{j:j\neq i}a_{ij} g(x^j-x^i)}
		we assume
		\eq{\label{Ionefinite}
			\E_{\pi}\Ione(\bfx,\bfa):=\E_{\pi}\frac{1}{N}\sum_{i=1}^N\psi^i(\bfx,\bfa)<C, 
		}
		and
		\eq{\label{Itwofinite}
			\E_{\pi}\Itwo(\bfx,\bfa):=\frac{1}{N}\E_{\pi}\sum_{i=1}^N \left[\psi^i(\bfx,\bfa)\right]^2< C ,}	
		where $C>0$ is a constant independent of $N$ (and the notation $\E_{\pi}$ has been introduced after Assumption \ref{H3new}). 
		\item { if $K$ satisfies \ref{H1} and is bounded then we assume that \eqref{Ionefinite} and \eqref{Itwofinite} are satisfied with  $g\equiv 1$. }
	\end{itemize}
	
\end{enumerate}

\begin{Note}\label{note:hypothesis}\textup{Hypothesis \ref{H1} is (more than) sufficient to ensure that system \eqref{slow} -\eqref{fast} has a strong solution. 
		Hypothesis \ref{H2} will be invoked to obtain uniform in time results. Hypothesis \ref{H3} and \ref{H3new} are assumptions on the initial conditions.  Note that if $g \equiv 1$ then \eqref{def:psii} boils down to assuming that at time $t=0$ each particle has a number of links which is on average of order one,  i.e.
		$$\frac1N \E_{\pi}\sum_{i=1}^N\sum_{j:j\neq i}A_{ij}(0)\leq C,$$
		where $C>0$ is a constant independent of $N$. In general one should think of $g(x)$ as being {$g(x)={\rm max}(1,|x|)$. In this case \eqref{Ionefinite} implies both that the average number of links per particle is of order $1$, and that there are not too many links between ``far away" particles.} {If \eqref{Ionefinite} is needed to control the growth in $N$ of the expected value of $\psi^i$, assumption \eqref{Itwofinite} is needed to control the growth in $N$ of the  second moment of $\psi^i$. }\\
		When $K$ is unbounded,  a crucial step in the  proof of Theorem \ref{theo:semigroup} consists of controlling  the probability of two far away particles being connected by a link. Lemma \ref{lem:Ione}  and Lemma \ref{lem:Itwo}  address this issue. 
		See also Note \ref{note:note45} and the comments after the statement of Lemma \ref{lem:Ione} on the significance of {assumption \ref{H3} }.}
\end{Note}

\section{Proof of Theorem \ref{thm:main1}}\label{sec:main}
Let us  consider the LHS of \eqref{main1_statement}, namely
\eq{
	E &:= \E \left( \int_{\R^n} u(y) \,  d \hat{\mu}_{X^{N,\varepsilon}_t}(y) \, - \, \int_{\R^n} u(y) \rho_t(y) dy  \right)^2 \\
	& = 
	\frac{1}{N^2} \E \sum_{i,j} u(X^{i,N, \varepsilon}_t) u(X^{j,N, \varepsilon}_t)  +\left(\int_{\R^n} u(y) \rho_t(y) dy\right)^2\\
	&-\frac{2}{N}\left(\int_{\R^n} u(y) \rho_t(y) dy\right)  \E \sum_i u(X^{i,N, \varepsilon}_t) \,.
}
We now split the above expression introducing terms which contain the dynamics of the particles which evolve according to the averaged dynamics \eqref{slow2}:
\[
E =E^{\rm av} +{E^{\rm part}}
\]
with
\eq{\label{Eav}
	E^{\rm av} & :=\frac{1}{N^2} \E\sum_{i,j} \left( u(X^{i,N,\ep}_t) u(X^{j,N,\ep}_t)-u(\bar{X}^{i,N}_t)u(\bar{X}^{j,N}_t)\right)\\
	&-\frac{2}{N}\left(\int u(y) \rho_t(y) dy\right) \E\left( \sum_i^N u(X^{i,N,\ep}_t) 
	- u(\bar{X}^{i,N}_t)\right)
}
and
\eq{\label{Ebar}
	{E^{\rm part}} & := \frac{1}{N^2} \E\sum_{i,j}^N u(\bar{X}^{i,N}_t)u(\bar{X}^{j,N}_t) 
	-\frac{2}{N}\E\left(\int u(y) \rho_t(y) dy\right) \E \sum_i^N u(\bar{X}^{i,N}_t) \\ 
	& +\left(\int u(y) \rho_t(y) dy\right)^2\\
	&= \E\lr{\frac1N\sum_{i}u(\bar X^{i,N}_t)-\int u(y) \rho_t(y) dy}^2,
}
where $\bar{X}^{i,N}_t$ satisfies \eqref{slow2}.
To obtain \eqref{main1_statement}, we need to show
\begin{equation}\label{EEavT}
\sup_{t\in[0,T]}\lv E^{\rm av} \rv    \leq  C(1+T)N\varepsilon.    
\end{equation}
and 
\begin{equation}\label{EEpart}
\sup_{t\in[0,T]}\lv E^{\rm part} \rv    \leq \frac{C}{N}e^{CT}.    
\end{equation}
To obtain \eqref{main0_statement} we need to prove, under the stated additional assumptions: 
{ 
	\begin{equation}\label{EEavunif}
	\sup_{t\geq 0}\lv E^{\rm av} \rv    \leq CN\ep    
	\end{equation}
	and 
	\begin{equation}\label{EEpartunif}
	\sup_{t\geq 0}\lv E^{\rm part} \rv    \leq \frac{C}{N}.    
	\end{equation}
}
The averaging estimates \eqref{EEavT} and \eqref{EEavunif}  are proven in Section \ref{Sec:Av};  the estimates \eqref{EEpart} and \eqref{EEpartunif} are studied in  Section \ref{Sec:MF}.  Once such estimates are shown, the proof is concluded. 

Note that, once the above estimates have been proved, the $\ep\rightarrow 0$ and $N\rightarrow \infty$ limits are taken {\em together}. This raises a natural question about the commutativity of these two limits. At least formally,  the limit   $\ep \rightarrow 0 $ first and $N \rightarrow \infty$ could be considered and would produce the non-linear Fokker-Planck equation \eqref{macro}. It is much less clear how, even formally,  one would let $N \rightarrow \infty$ first in this sparse graph regime (and indeed, as mentioned in the introduction, the works \cite{Degond16,BDZ,BCDPZ,Taylor17}, which contain a formal derivation by taking the $N$ limit first, do need to make propagation-of-chaos-type ansatzes to carry out the procedure). Note indeed that, once the $\ep$ limit is taken, i.e. once the dynamics \eqref{slow2} is obtained, our choice of scaling of $\nu_f^N$ \eqref{Nscaling} makes \eqref{slow2} treatable through mean-field arguments, as anticipated in \cite{BDPZ}.  Either way, unless the two limits are taken together, the analysis can never be revealing of the correct critical regime $\ep N\ll 1$ to be considered.

\section{The Averaging limit: proof of \eqref{EEavT} and \eqref{EEavunif}}\label{Sec:Av}
Using the exchangeability of the particles we can rewrite \eqref{Eav} as
\eq{\label{eq:3.1}
	E^{\rm av}
	& =
	\frac{1}{N}\mathbb E\left[u^2(X^{1,N, \ep}_t) -u^2(\bar{X}_t^{1,N\ep})\right]
	\\
	&+\frac{N-1}{N}
	\left(\mathbb{E}\left[u(X_t^{1,N\ep})u(X_t^{2,N\ep})\right]
	-\mathbb{E}\left[u(\bar{X}_t^{1,N})u(\bar{X}_t^{2,N})\right]\right) \\
	&-2\left(\int u(y) \rho_t(y) dy\right)\mathbb{E}\left[u(X^{1,N,\ep}_t) 
	-u(\bar{X}^{1,N}_t)\right]\\
	& \leq  \frac{C}{N}  +\frac{N-1}{N} \E_{\pi}\lv (\cP^{N,\varepsilon}_t u^{(2)})(\bfx,\bfa) -(\bar{\cP}^N_t u^{(2)})(\bfx,\bfa)\rv \\
	& +C \E_{\pi}\lv(\cP_t^{N,\varepsilon} u^{(1)})(\bfx,\bfa) - (\bar{\cP}^N_t u^{(1)})(\bfx,\bfa)\rv, 
}
where $u^{(k)}:\R^{Nn}\to \R$ is defined as $u^{(k)}(\bfx):=u(x^1)\ldots u(x^k)$;  in the last inequality we have used the boundedness of $u$ and the fact that, by independence of the sources of noise, $\E=\E_{\pi}\E_B$. 

The result now follows once we can control the difference between the semigroups
$\cP^{N,\varepsilon}_t$ and $\bar{\cP}^N_t$. This is the purpose of the following theorem, which is the main result of this section.  Before stating and proving such a theorem let us emphasize again that throughout this section we will first consider a fixed initial datum $X_0^{N, \ep}=\bfx, A^{N, \ep}(0)=\bfa$ and then at the latest possible moment we will take expectation $\E_{\pi}$ with respect to the noise in the initial data.  
\begin{theorem} \label{theo:semigroup}
	With the notation introduced so far, suppose Assumption \ref{H3new}, Hypothesis \ref{H1} and Hypothesis \ref{H3} are satisfied. 
	Then there exists some $\varepsilon_0$ such that for any {$t\in[0,\infty)$}, and any $N,\varepsilon$ with $N\varepsilon<\varepsilon_0$ we have 	
	\eqh{
		\E_{\pi}\lvert \cP_t^{N,\varepsilon} f(\bfx, \bfa)\!-\!\cPbar_tf(\bfx)\rvert
		\!\leq\! \varepsilon NC (1+t) \E_{\pi}\!\left(C_1 + C_3 \Ione(\bfx, \bfa)e^{\frac{-C_2t}{\varepsilon}}+ \Itwo(\bfx, \bfa)e^{\frac{-C_4t}{\varepsilon}} \right)
	}
	for every $f\in C_b^2(\R^{nN})$ (where {the constant $C$ does not depend on $t$}, $N$ or $\ep$ but it will depend on a suitable norm of $f$, see \eqref{defnorm}).  
	Moreover, if Hypothesis \ref{H2} holds for some $\kappa_3\geq \kappa_{av}$ (the latter being defined in Lemma \ref{lem:derivativeest}),  then the above estimate is uniform in time, i.e.
	\eqh{
		\E_{\pi}\lvert \cP_t^{N,\varepsilon} f(\bfx, \bfa)-\cPbar_tf(\bfx)\rvert
		\leq \varepsilon NC \E_{\pi} \left(C_1 + C_3\Ione(\bfx, \bfa)e^{-C_2\frac{t}{\varepsilon}}+ \Itwo(\bfx, \bfa)e^{-C_4\frac{t}{\varepsilon}} \right)\,.
	}
\end{theorem}
Before proving Theorem \ref{theo:semigroup} we state three lemmata, which are key to proving Theorem \ref{theo:semigroup}. We will prove such lemmata after the proof of Theorem \ref{theo:semigroup}. 
The first two lemmata, Lemma \ref{lem:Ione} and Lemma \ref{lem:Itwo}, let us express more precisely the idea that the average number of links per particle remains of order $1$, if such a property is true at time zero. Lemma \ref{lem:derivativeest} contains bounds on the derivatives of the semigroup $\cPbar_t$, which are key to obtaining uniform in time results.

\begin{lemma}\label{lem:Ione}
	With the notation introduced so far, suppose Hypotheses \ref{H1} and \ref{H3} hold.
	If $N$ and $\ep$ are such that the product $N\varepsilon$ is sufficiently small, {then for any $t\in[0,\infty)$} there exist positive constants $C_1,C_2$ {independent of $t$} such that
	\begin{align*}
	(\cP_t^{N,\ep} \Ione\!) (\bfx,\bfa)\!= \!\mathbb{E}_B\!\!\left[\Ione\!(X_t^{N, \ep},\!A^{N, \ep}(t)) \vert (X_0^{N, \ep},A^{N, \ep}(0))\!=\!(\bfx, \bfa)\right] \!\!\leq\! C_1\!+\!\Ione\!(\bfx,\bfa)e^{-\frac{C_2t}{\varepsilon}}\!.
	\end{align*}
\end{lemma}
According to Lemma \ref{lem:Ione}, the quantity $\mathbb{E}_B[\Ione(X_t^{N, \ep},A^{N, \ep}(t))]$ remains of order one (in $N$) uniformly in time if it is of order one at the initial time. However Lemma \ref{lem:Ione} says also something more: because we are working in the regime $N\ep \ll 1$, so e.g. $\ep= 1/N^{\alpha}$ for some $\alpha>1$, if $\Ione$ is large in $N$ at time zero, then such initially large value quickly decays in expectation, due to the factor  $e^{-C_2t/\varepsilon}$.  Therefore it is reasonable to assume that the average number of links in the initial configuration is $O(1)$ in $N$.
Lemma \ref{lem:Itwo} below strengthens this result by giving some control on the second moment as well. 
\begin{lemma}\label{lem:Itwo}
	With the same assumptions and notation as in Lemma \ref{lem:Ione}, there exists $\varepsilon_0>0$ such that if $N\varepsilon< \varepsilon_0$ then, {for any $t\in[0,\infty)$ and for some positive constants $C_1,C_2,C_3,C_4>0$ independent of $t$}, we have
	\begin{equation*}
	\cP_t^{N,\varepsilon}\Itwo(\bfx,\bfa) \leq C_1 + C_3\Ione(\bfx,\bfa)e^{-C_2\frac{t}{\varepsilon}}+ \Itwo(\bfx,\bfa)e^{-C_4\frac{t}{\varepsilon}}.
	\end{equation*}
	Here $\Ione$ is defined as in Lemma \ref{lem:Ione}.
\end{lemma}
For $f:\R^{Nn}\to\R$, $f \in C_b^2(\R^{nN})$, we define the norm 
\begin{align}
\lVert f\rVert_{C_b^2(\R^{nN})}^2&:= \lVert f\rVert_\infty^2+ \sum_{i=1}^N \lVert \nabla_i f\rVert_\infty^2 + \sum_{i,j=1}^N \lVert \nabla^2_{ij} f\rVert_\infty^2,\label{defnorm}
\end{align} 
{and seminorm
	\begin{align}    
	[[f]]_{C_b^2(\R^{nN})}^2 &:=\sup_{\bfx \in \R^{nN}}\left(\sum_{i=1}^N\lvert \nabla_{i} f(\bfx) \rvert^2+ \sum_{i,j=1}^N \lvert \nabla^2_{ij} f(\bfx)\rvert^2 \right)\label{defseminorm}
	\end{align}}
where $\nabla_i f$ is the gradient with respect to the variable $x^i\in \R^n$, $x^i=(x^i_{1} \dd x^i_n)$, so that $\|\nabla_i f\|_{\infty}:= \max_{k \in \{1 \dd n\} }
\|(\nabla_i f)_k \|_{\infty} $, where $(\nabla_i f)_k$ is the $k$-th component of the gradient $\nabla_i$ and  
$\|\cdot\|_\infty$ is the supremum norm in $\R^{Nn}$. Analogously, $\nabla_{ij}^2f$ is the $n \times n$ matrix containing all the second derivatives with respect to the variables $x^i, x^j \in \R^n$. 
Let us emphasize that the dependence on $n$ of this quantity is irrelevant to our purposes, as $n$ is fixed; what matters here is the dependence on $N$ so, to simplify notation, the proofs of all the lemmata will be done in the case $n=1$. With this notation in place we can state the next lemma. 
\begin{lemma}\label{lem:derivativeest}
	Suppose  Hypothesis \ref{H1} holds. 
	Recall the function $\Kbar: \R^n \rightarrow \R^n$ is defined in \eqref{eq:Kbardef} and note that, as a consequence of our assumptions on $K$ and $\varphi_R$,  there exist positive constants $\kappa_1,\kappa_2$ such that $\|\nabla\Kbar(z)\|_{\infty}\leq \kappa_1/N$ and $\|\nabla^2\Kbar(z)\|_{\infty}\leq \kappa_2/N$. \footnote{Here $\nabla$ and $\nabla^2$ are gradient and Hessian with respect to $z \in \R^n$.} 
	Let $\kappa_3 \in \R,\kappa_4>0$ be  constants such that  ${\rm Hess}V(z)\geq \kappa_3 Id_n$  and $\|\nabla^3 V(z)\|_{\infty}\leq \kappa_4$ for all $z\in \R^n$.{\footnote{$\nabla^3$ denotes all the third derivatives with respect to $z \in \R^n$.}} 
	Then
	\begin{description}
		\item[i)] There is a constant $C>0$ which depends only on $\kappa_1,\kappa_2,\kappa_3$ and $\kappa_4$ such that for all $N\in \mathbb{N}$, $f\in C_b^2(\R^{nN})$ and $t\geq 0$ we have
		\begin{equation}\label{eq:C2toC2ineq}
		{[[\cPbar_t f]]_{C_b^2(\R^{nN})} \leq C \lVert f\rVert_{C_b^2(\R^{Nn})}.}
		\end{equation}
		
		\item[ii)] Moreover, if \ref{H2} holds with $\kappa_3$ such that, for some $\delta>0$, 
		\begin{equation*}
		{ \kappa_3 \geq \kappa_{\rm av}:=\frac\delta2+2\kappa_1+\frac{\kappa_4+4\kappa_2}{\delta +4\kappa_2+\kappa_4+4\kappa_1}},
		\end{equation*}
		then,    
		\begin{align}\label{eq:expdecayofderivative}
		{ [[\cPbar_t f]]_{C_b^2(\R^{nN})}  \leq Ce^{-\delta t} \|f\|_{C^2_b(\R^{Nn})}^{{2}} }\,
		\end{align}
		for some $C>0$ independent of {$N$ and $t$}.
	\end{description}
\end{lemma}

\subsection{Heuristics: formal expansion of the semigroup $\cP_t^{N, \ep}$}
To motivate the structure of the proof of Theorem \ref{theo:semigroup}, we first present a heuristic computation. 
One can consider formally expanding the semigroup $\cP_t^{N,\varepsilon}$ in  powers of $\varepsilon$ (this procedure is similar in spirit to the procedure presented in \cite{PS}, the only difference is that here we expand the semigroup, while in \cite{PS} the  emphasis is on the expansion of the generator):
\begin{equation}\label{AAA}
(\cP_t^{N,\varepsilon} f)(\bfx,\bfa)=f^0_t(\bfx,\bfa)+\varepsilon f_t^1(\bfx,\bfa)+...
\end{equation}
Using that $\cP_t^{N,\varepsilon}$ solves the PDE \eqref{eq:epsPDE}, i.e.
$$
\pa_t \cP_t^{N,\varepsilon} f- \cL^{N, \ep} \cP_t^{N,\varepsilon} f=0 \, ,
$$
and recalling that the generator $\cL^{N, \ep}$ of $\cP_t^{N, \ep}$ is given by \eqref{gen:Lep}, following \cite{PS}, 
we can insert \eqref{AAA} in the above  and, comparing terms containing the same power of $\ep$ we deduce:
\begin{align}
O\lr{\frac{1}{\varepsilon}}:& &\cL_F f_t^0=0\label{eq:orderminus1}\\
O(1):& &\partial_t f_t^0-\cL_Sf_t^0=\cL_Ff_t^1\label{eq:order0}
\end{align}
As we have explained, if we fix $X^{N,\ep}_t$ then the dynamics of $A^{N, \ep}(t)$ is ergodic (i.e. it admits a unique invariant measure). This implies that the only solutions $u(\bfx,\bfa)$ to $\cL_F u=0$ are constant in $\bfa$ 
\footnote{If $\cP_t^{A,x}$ denotes the semigroup associated to $A$ with the $x$-coefficient frozen. Then let $\mu^x$ denote the unique invariant measure of the frozen system, note that $\cP_t^{A,x}f$ converges to $\int f d\mu^x$. Now 
	$$
	\cP_t^{A,x}f(A)=f(A)+\int_0^t \cP_s^{A,x}\cL_Ff(A) ds.
	$$
	Therefore, if $f$ is a solution to $\cL_Ff=0$ then $\cP_t^{A,x}f(A)=f(A)$ and letting $t$ tend to $\infty$ we must have that $f=\mu^x(f)$, i.e. $f$ is a constant function in the variable $A$.}
so we have $f_t^0(\bfx,\bfa)=f_t^0(\bfx)$ from \eqref{eq:orderminus1}. Now integrating \eqref{eq:order0} with respect to $\mu_\bfx$, the unique invariant measure on $\ta$ with the $\bfx$-coefficient frozen, we obtain
\begin{equation}\label{ded}
\partial_tf_t^0(\bfx)-\bar{\cL}^N f_t^0(\bfx)=0 \, ,
\end{equation}
where we recall that, from \eqref{fbari} and \eqref{10star},  $\bar\cL^N$ is given by
\begin{align*}
\bar{\cL}^Nf(\bfx)
&=\sum_{i=1}^N\left[\lr{-\nabla_i V(x^i)+\sum_{j\neq i}
	\frac{\nu_f^N\varphi_R(|x^i-x^j|)}{\nu_f^N\varphi_R(|x^i-x^j|)+\nu_d} K(x^j-x^i)} \nabla_{i}f(\bfx)\right]\\
&+D {\rm Tr}(\rm{Hess}f)(\bfx).
\end{align*}
From \eqref{ded} one must then have $f_t^0(\bfx)=(\bar\cP_t^N f)(\bfx)$, hence, from \eqref{AAA},
\begin{equation}\label{BBBB}
(\cP_t^{N,\varepsilon} f)(\bfx,\bfa)= (\bar\cP_t^N f)(\bfx)+\varepsilon f_t^1(\bfx,\bfa)+...
\end{equation}
One  can now solve\footnote{{ This solution was found by first considering the case when $N=2$ in which \eqref{eq:order0} reduces to  two linear equations, then proposing an ansatz based on this solution and verifying that it does indeed solve \eqref{eq:order0}.
}} \eqref{eq:order0} to find $f_t^1$: 
\begin{equation*}
f_t^1(\bfx,\bfa) = \sum_{i=1}^N\sum_{j\neq i} \frac{1}{\nu_f^N\varphi_R(|x^i-x^j|)+\nu_d}   K(x^j- x^i) a_{ij} \partial_i\bar\cP_t^Nf(\bfx).
\end{equation*}
Note that the solution to \eqref{eq:order0} is {\em not} unique\footnote{ {This solution is not unique since $\cL_F$ has a non-trivial kernel and in particular any function which does not depend on $\bfa$ belongs to the kernel; this could be resolved by suitable boundary conditions but for our purposes we do not require the solution to be unique and will just work with the solution found in \eqref{eq:ft1def}.}}, however the above is one possible solution, and the one with which we will work. 
To simplify notation, define 
\begin{equation}\label{Ktilde}
\tilde{K}(z) = \frac{1}{\nu_f^N\varphi_R(z)+\nu_d}   K(z). \end{equation}
(We should  denote $\tilde K$ by $\tilde K^N$, we don't do so to avoid cumbersome formulas later on).
Now we can write
\begin{equation}\label{eq:ft1def}
\fone (\bfx,\bfa) = \sum_{i=1}^N\sum_{j\neq i} \tilde{K}(x^j- x^i) a_{ij} \partial_i\bar\cP_t^Nf(\bfx).
\end{equation}
We shall later use this explicit expression for $\fone$. Before moving on to the actual proof of Theorem \ref{theo:semigroup}, let us make a remark.
\begin{Note}\label{note:note45}\textup{A few comments on the definition of $f_t^1$. 
		\begin{itemize}
			\item It is clear from \eqref{BBBB} that, in order to control the difference $\cP_t^{N, \ep} -\bar\cP_t^N$, we need to control $f_t^1$; comparing the above expression for $f_t^1$ with the one for $\Ione$, one can see that the need to control $f_t^1$ is where Lemma \ref{lem:Ione} and Lemma \ref{lem:Itwo} stem from and, in turn, the reason for imposing Hypothesis \ref{H3} on the initial data.   Note that if $K$ is bounded (with bounded derivatives) then $\tilde{K}$ is also bounded (with bounded derivatives). This  relates to why \ref{H3} includes the two different cases of $K$ being bounded/unbounded. 
			\item If $f_t^1$ is as in \eqref{eq:ft1def}, then, as opposed to what what it may look like by observing the heuristic expansion \eqref{AAA}, $\ep f_t^1$ is not the whole term of order $\ep$ when we look at the difference between $\cP_t^{N, \ep}$ and $\bar\cP_t^{N, \ep}$ (cf \eqref{AAA} and \eqref{BBBB}). This is in no way   contradictory, because of the non-uniqueness of $f_t^1$. The whole term of order $\ep$ will appear during the rigorous proof, see \eqref{BB} and \eqref{eq:estimateonrwithV} below. 
		\end{itemize}
	}
\end{Note}

We now have all the tools to prove Theorem \ref{theo:semigroup}.
\subsection{Proof of Theorem \ref{theo:semigroup}} \label{sec:proofdetails_theorem}

\begin{proof}[Proof of Theorem \ref{theo:semigroup}]
	To simplify notation we set the diffusion coefficient $D=1$ and $n=1$, so that $f \in C_b^2(\R^N)$ and $K,V:\R \rightarrow \R$. We will therefore write $\partial_i$ instead of $\nabla_i$ and $V'$ denotes derivative of $V$ with respect to its argument. Let $f_t^1$ be the function  defined in \eqref{eq:ft1def} and observe 
	that, by a direct calculation,  such a function is a solution to
	\begin{equation}\label{eq:fonePDE}
	\bar\cL^N \cPbar_t f-\cL_S\cPbar_tf=\cL_F\fone.
	\end{equation}
	Motivated by \eqref{AAA} and \eqref{BBBB}, we set
	\begin{equation}\label{BB}
	r_t^{N,\varepsilon}(\bfx,\bfa):=(\cP_t^{N,\varepsilon} f)(\bfx,\bfa)-(\bar\cP_t^Nf)(\bfx)-\varepsilon f_t^1(\bfx,\bfa) \,.
	\end{equation}
	Differentiating $r_t^{N,\varepsilon}$ with respect to time gives
	\begin{equation*}
	\partial_tr_t^{N,\varepsilon}(\bfx,\bfa)=\partial_t(\cP_t^{N,\varepsilon} f)(\bfx,\bfa)-\partial_t(\bar\cP_t^Nf)(\bfx)-\varepsilon \partial_tf_t^1(\bfx,\bfa).
	\end{equation*}
	Using \eqref{eq:epsPDE} we have
	\begin{align*}
	\partial_t r_t^{N,\varepsilon}(\bfx,\bfa) \!& =\!
	\cL^{N,\varepsilon}\cP_t^{N,\varepsilon} f(\bfx,\bfa)-\partial_t\bar\cP_t^N f(\bfx)-\varepsilon \partial_tf_t^1(\bfx,\bfa)\\
	& \!\!\!\!\!\stackrel{\eqref{BB}}{=}\!\cL^{N,\varepsilon} r_t^{N,\varepsilon}(\bfx,\bfa)\!+\!
	\cL^{N,\varepsilon}\bar\cP_t^N f(\bfx) \!+\!\varepsilon \cL^{N,\varepsilon} f_t^1(\bfx,\bfa)\! -\!\partial_t\bar\cP_t^N f(\bfx) \!- \!\varepsilon \partial_tf_t^1(\bfx,\bfa).
	\end{align*}	
	Since $\bar\cP^N_tf(\bfx)$ does not depend on $\bfa$ we have $\cL^{N,\varepsilon}\bar\cP_t^N f(\bfx)=\cL_S\bar\cP_t^N f(\bfx)$ and using \eqref{eq:fonePDE} we obtain
	\begin{align*}
	\partial_t r_t^{N,\varepsilon}(\bfx,\bfa)&=\cL^{N,\varepsilon} r_t^{N,\varepsilon}(\bfx,\bfa)-\cL_Ff_t^1(\bfx,\bfa)+\varepsilon \cL^{N,\varepsilon} f_t^1(\bfx,\bfa) -\varepsilon \partial_tf_t^1(\bfx,\bfa)\\
	&=\cL^{N,\varepsilon} r_t^{N,\varepsilon}(\bfx,\bfa)+\varepsilon \left(\cL_S f_t^1(\bfx,\bfa) -\partial_tf_t^1(\bfx,\bfa)\right).
	\end{align*}
	Therefore, using the variation of constants formula, we have
	\begin{align*}
	r_t^{N,\varepsilon}(\bfx,\bfa) &= \cP_t^{N,\varepsilon} r_0^{N,\varepsilon}(\bfx,\bfa)+ \varepsilon\int_0^t \cP_{t-s}^{N,\varepsilon} \left(\cL_S f_s^1 -\partial_sf_s^1\right)(\bfx,\bfa)ds.
	\end{align*}
	The remainder of the proof is divided into 3 steps:
	\begin{itemize}
		\item Step 1: Show that 
		\begin{equation}\label{S1}
		\lvert \cP_t^{N,\varepsilon} r_0^{N,\varepsilon}\rvert \leq CN\varepsilon\left(C_1+\Ione(\bfx,\bfa)e^{-C_2\frac{t}{\varepsilon}}\right) \lVert f\rVert_{C_b^2(\R^N)} 
		\end{equation}
		\item Step 2: Show that
		$$
		\lvert\cP_{t-s}^{N,\varepsilon}(\cL_S-\partial_s)f_s^1\rvert\!\leq\! CN\!\left(\!C_1\! +\! C_3\Ione(\bfx,\bfa)e^{-C_2\frac{t}{\varepsilon}}\!+\! \Itwo(\bfx,\bfa)e^{-C_4\frac{t}{\varepsilon}} \right)\!\lVert f\rVert_{C_b^2(\R^N)}.
		$$
		Moreover, if ii) of Lemma \ref{lem:derivativeest} is satisfied, 
		$$
		\!\!\!\!\!\!\lvert\cP_{t-s}^{N,\varepsilon}(\cL_S-\partial_s)f_s^1\rvert\!\leq\! CNe^{-\delta s}\!\left(\!C_1\! +\! C_3\Ione(\bfx,\bfa)e^{\frac{-C_2t}{\varepsilon}}\!+\! \Itwo(\bfx,\bfa)e^{\frac{-C_4t}{\varepsilon}} \right)\!\lVert f\rVert_{C_b^2(\R^N)}.
		$$
		for some $\delta>0$.
		\item Step 3: Show that
		\begin{equation}\label{S3}
		\lvert\fone (\bfx,\bfa)\rvert\leq CN \Ione(\bfx,\bfa)\lVert f\rVert_{C_b^2(\R^N)}.
		\end{equation}
	\end{itemize}
	{The above estimates are valid for any $t\in[0,\infty)$ and for constants $C,C_1,...,C_4$ that are independent of $t$.}
	This completes the proof since Step 1 and Step 2 imply
	\begin{equation}\label{eq:estimateonr}
	\lvert r_t^{N,\varepsilon}(\bfx,\bfa)\rvert\!\leq\! C\varepsilon N(1+t)\left(C_1 + (1+C_3)\Ione(\bfx,\bfa)e^{-C_2\frac{t}{\varepsilon}}+ \Itwo(\bfx,\bfa)e^{-C_4\frac{t}{\varepsilon}} \right)\lVert f \rVert_{C_b^2(\R^N)}.
	\end{equation}
	Then, by \eqref{eq:estimateonr}, \eqref{BB} and Step 3, we have
	\begin{align*}
	&\left\lvert \cP_t^{N,\varepsilon} f(\bfx,\bfa)-\cPbar_t f(\bfx)\right\rvert \\
	&=\lvert \varepsilon\fone(\bfx,\bfa) + r_t^{N,\varepsilon}(\bfx,\bfa)\rvert\\
	&\leq CN\varepsilon (1+t) \left(C_1 + (1+C_3)\Ione(\bfx,\bfa)e^{-C_2\frac{t}{\varepsilon}}+ \Itwo(\bfx,\bfa)e^{-C_4\frac{t}{\varepsilon}} \right)\lVert f\rVert_{C_b^2(\R^N)}.
	\end{align*}
	Moreover, if ii) of Lemma \ref{lem:derivativeest} is satisfied,  from Step 1 and Step 2 we also have
	\begin{equation}\label{eq:estimateonrwithV}
	\lvert r_t^{N,\varepsilon}(\bfx,\bfa)\rvert\leq \frac{C\varepsilon N}{\delta}\left(C_1 + (1+C_3)\Ione(\bfx,\bfa)e^{-C_2\frac{t}{\varepsilon}}+ \Itwo(\bfx,\bfa)e^{-C_4\frac{t}{\varepsilon}} \right)\lVert f \rVert_{C_b^2(\R^N)}.
	\end{equation}
	Hence, 
	\begin{align*}
	\!\!\!\left\lvert \cP_t^{N,\varepsilon} f(\bfx,\bfa)\!-\!\cPbar_tf(\bfx)\right\rvert\!\leq\! CN\varepsilon\!\left(\!C_1 \!+\! (1\!+\!C_3)\Ione(\bfx,\bfa)e^{\frac{-C_2t}{\varepsilon}}\!\!+\! \Itwo(\bfx,\bfa)e^{\frac{-C_4t}{\varepsilon}} \!\right)\!  \lVert f\rVert_{C_b^2(\R^N)}.
	\end{align*}
	To prove Step 1, note that from \eqref{eq:epsPDE} we have  $r_0^{N,\varepsilon}(\bfx,\bfa)=f(\bfx)-f(\bfx)-\varepsilon f_0^1(\bfx,\bfa)=-\varepsilon f_0^1(\bfx, \bfa)$. From \eqref{eq:ft1def} and the definition of $\psi^i$ (see \eqref{def:psii}) we get
	\begin{align}\label{eq:Ptf01}
	\lvert \cP_t^{N,\varepsilon} r_0^{N,\varepsilon}(\bfx,\bfa)\rvert &=\varepsilon\lvert\cP_{t}^{N,\varepsilon} f_0^1(\bfx,\bfa)\rvert \leq  C\varepsilon \sum_{i=1}^N (\cP_t^{N,\varepsilon}\psi^i)(\bfx,\bfa) \lVert \partial_if\rVert_\infty \\
	&\stackrel{\eqref{defnorm}}{\leq} CN\varepsilon\lVert f\rVert_{C_b^2(\R^N)}(\cP_t^{N,\varepsilon}\Ione)(\bfx,\bfa).
	\end{align}
	Here $C$ is a positive constant that varies line by line but is independent of $N,\varepsilon$ and $\bfa$. By Lemma \ref{lem:Ione} and the above one obtains \eqref{S1}, 
	hence we have proven Step 1.
	
	To prove Step 2, we apply the operator $(\cL_S-\partial_s)$ to the definition of $f_s^1$ (see \eqref{eq:ft1def}).
	\begin{align} \label{DD}
	(\cL_S-\partial_s)f_s^1 & = \sum_{i=1}^N\sum_{j\neq i}(\cL_S-\partial_s)  \left( a_{ij}\tilde{K}(x^j- x^i) \partial_i (\cPbar_sf)(\bfx)\right)\nonumber\\
	&=\sum_{i=1}^N\sum_{j\neq i}  \cL_S( \tilde{K}(x^j- x^i)) a_{ij} \partial_i\cPbar_sf(\bfx) \\
	&+ 2\sum_{i=1}^N\sum_{j\neq i}\sum_{k=1}^N \partial_k( \tilde{K}(x^j- x^i)) a_{ij} \partial_{k}\partial_i\cPbar_sf(\bfx) \nonumber\\
	&+\sum_{i=1}^N\sum_{j\neq i} \tilde{K}(x^j- x^i) a_{ij} (\cL_S-\partial_s)  \left(\partial_i\cPbar_sf(\bfx)\right)
	\end{align}
	We study the three addends on the RHS of the above; to this end, we start by considering $\cL_S ( \tilde{K}(x^j- x^i))$ (contained in the first addend): 
	\begin{align*}
	\cL_S ( \tilde{K}(x^j\!-\! x^i))\! &= \sum_{k=1}^N\sum_{\ell\neq k} K(x^\ell-x^k)a_{k\ell}\partial_{k}( \tilde{K}(x^j- x^i)) + \sum_{k=1}^N \partial_{kk}^2 (\tilde{K}(x^j- x^i))\\
	&-\sum_{k=1}^N V'(x^k)\partial_k( \tilde{K}(x^j- x^i))\\
	&=-\sum_{\ell\neq i} K(x^\ell-x^i)a_{i\ell} \tilde{K}'(x^j- x^i) + \tilde{K}''(x^j- x^i)+ \tilde{K}''(x^j- x^i)\\
	&+\!\sum_{\ell\neq j}\! K(x^\ell\!-\!x^j)a_{j\ell}\tilde{K}'(x^j\!-\! x^i)
	\!+\!V'(x^i) \tilde{K}'(x^j- x^i)\!-\! V'(x^j) \tilde{K}'(x^j\!-\! x^i)\\
	&\leq C(1+\sum_{\ell\neq i} a_{i\ell}\lvert K(x^\ell-x^i)\rvert+\sum_{\ell\neq j} a_{j\ell}\lvert K(x^j-x^\ell)\rvert )+C\lvert x^i -x^j\rvert.
	\end{align*}
	To obtain the last line we have used that $V'$ is Lipschitz and that $K'$ and $K''$ are bounded (hence the boundedness of $\tilde K'$ and $\tilde{K}''$). Recalling  that $g$ is as defined in \ref{H3}, we then have
	\begin{equation}\label{C1}
	\cL_S ( \tilde{K}(x^j- x^i)) 
	\leq C(1+\psi^i(\bfx,\bfa)+\psi^j(\bfx,\bfa) )+C g(x^i -x^j).
	\end{equation}
	Next, we  consider the term $(\cL_S-\partial_s)  \left(\partial_i\cPbar_sf(\bfx)\right)$ (contained in the third addend of \eqref{DD}). 
	Note first that, from \eqref{eq:Kbardef} and \eqref{Ktilde}, we have 
	\begin{equation}\label{L}
	\bar{K}^N(z)= \nu_f^N \, \varphi_R(z) \tilde{K}(z).
	\end{equation}
	Hence,
	\begin{align*}
	&(\cL_S-\partial_s)  \left(\partial_i\cPbar_sf(\bfx)\right)\\
	& = (\cL_S\partial_i-\partial_i\overline{\cL}^N)  \left(\cPbar_sf(\bfx)\right)\\
	&=\sum_{k=1}^N\sum_{\ell\neq k} K(x^\ell-x^k) a_{k\ell}\partial_k\partial_i{\cPbar}_sf(\bfx) -\sum_{k=1}^NV'(x^k)\partial_k\partial_i\cPbar_sf(\bfx)\\
	&-\sum_{k=1}^N\sum_{\ell\neq k} \tilde{K}(x^\ell-x^k) \nu_f^N\varphi_R(x^\ell-x^k)\partial_i\partial_k{\cPbar}_sf(\bfx) + \sum_{k=1}^NV'(x^k)\partial_i\partial_k\cPbar_sf(\bfx)\\
	&-\sum_{k=1}^N\sum_{\ell\neq k}  \nu_f^N\partial_i(\tilde{K}(x^\ell-x^k)\varphi_R(x^\ell-x^k))\partial_k{\cPbar}_sf(\bfx)+ V''(x^i)\partial_i\cPbar_sf(\bfx)\, ,
	\end{align*}
	so that
	\eq{\label{C2}
		&(\cL_S\!-\!\partial_s)\!  \left(\partial_i\cPbar_sf(\bfx)\right)\!\\
		& \leq C\sum_{k=1}^N\sum_{\ell\neq k}  a_{k\ell}\lvert K(x^\ell-x^k)\rvert\lvert\partial_k\partial_i{\cPbar}_sf(\bfx)\rvert +C\sum_{k=1}^N(N-1) \nu_f^N\lvert\partial_i\partial_k\overline{\cP}^N_sf(\bfx)\rvert\\
		&+\sum_{\ell\neq i}  \nu_f^N(\varphi_R\tilde{K})'(x^\ell-x^i)\partial_i\overline{\cP}^N_sf(\bfx)-\sum_{k\neq i} \nu_f^N(\varphi_R\tilde{K})'(x^i-x^k)\partial_{{k}}\overline{\cP}^N_sf(\bfx){+C\lvert \partial_i{\cPbar}_sf(\bfx)\rvert}\\
		&\leq C\sum_{k=1}^N\sum_{\ell\neq k}  a_{k\ell}\lvert K(x^\ell-x^k)\rvert\lvert\partial_k\partial_i\overline{\cP}^N_sf(\bfx)\rvert +C\sum_{k=1}^N(N-1) \nu_f^N\lvert\partial_i\partial_k\overline{\cP}^N_sf(\bfx)\rvert \\
		&+C(N-1) \nu_f^N\lvert\partial_i\overline{\cP}^N_sf(\bfx)\rvert
		{+C \nu_f^N\sum_{k\neq i}\lvert\partial_k\overline{\cP}^N_sf(\bfx)\rvert }{+C\lvert \partial_i{\cPbar}_sf(\bfx)\rvert}\\
		&\leq C\sum_{k=1}^N\psi^k(\bfx,\bfa)\lvert\partial_k\partial_i\overline{\cP}^N_sf(\bfx)\rvert +C\sum_{k=1}^N\lvert\partial_i\partial_k\overline{\cP}^N_sf(\bfx)\rvert+C\lvert\partial_i\overline{\cP}^N_sf(\bfx)\rvert 
		{+C \nu_f^N\sum_{k\neq i}\lvert\partial_k\overline{\cP}^N_sf(\bfx)\rvert }\,. 
	}
	From \eqref{DD} we then have
	\begin{align*}
	(\cL_S-\partial_s)f_s^1  & {\leq}\sum_{i=1}^N\sum_{j\neq i} {\lvert} \cL_S( \tilde{K}(x^j- x^i)){\rvert} a_{ij} {\lvert\partial_i\cPbar_sf(\bfx)\rvert} \\
	&+ {2}\sum_{i=1}^N\sum_{j\neq i} a_{ij} \lvert \tilde{K}'(x^j- x^i)\rvert  \, {\lvert\partial_j\partial_i\cPbar_sf(\bfx)\rvert} 
	{+2 \sum_{i=1}^N\sum_{j\neq i} a_{ij} \lvert \tilde{K}'(x^j- x^i)\rvert} \, {\lvert\partial_i\partial_i\cPbar_sf(\bfx)\rvert}\\
	&+\sum_{i=1}^N\psi^i(\bfx,\bfa)\lvert(\cL_S-\partial_s)  \left(\partial_i\cPbar_sf(\bfx)\right)\rvert\\
	&\!\!\!\!\!\!\!\!\stackrel{\eqref{C1}, \eqref{C2}}{\leq} \sum_{i=1}^N\sum_{j\neq i}  C\left(1+\psi^i(\bfx,\bfa)+\psi^j(\bfx,\bfa) + g(x^i -x^j)\right) a_{ij} {\lvert\partial_i\cPbar_sf(\bfx)\rvert} \\
	&+ {2 C\sum_{i=1}^N \sum_{j\neq i} a_{ij}}{\lvert\partial_j\partial_i\cPbar_sf(\bfx)\rvert}\!+ 
	{2 C\sum_{i=1}^N \psi^i(\bfx,\bfa)}{\lvert\partial_i\partial_i\cPbar_sf(\bfx)\rvert} \\
	&+\!C\!\sum_{i=1}^N\psi^i(\bfx,\bfa)\sum_{k=1}^N\psi^k(\bfx,\bfa){\lvert\partial_k\partial_i\overline{\cP}^N_sf(\bfx)\rvert} 
	+C\sum_{i=1}^N\psi^i(\bfx,\bfa)\sum_{k=1}^N{\lvert\partial_i\partial_k\overline{\cP}^N_sf(\bfx)\rvert} \\
	& +C\sum_{i=1}^N\psi^i(\bfx,\bfa){\lvert\partial_i\overline{\cP}^N_sf(\bfx)\rvert} {+C\nu_f^N\sum_{i=1}^N \psi^i(\bfx,\bfa)\sum_{k\neq i} \lvert\partial_k\overline{\cP}^N_sf(\bfx)\rvert }\,. 
	\end{align*}
	{Now we want to eliminate in the right hand side above all derivatives of $\cPbar_sf$ in favor of $[[\cPbar_sf]]_{C_b^2(\R^{N})}$, and all terms involving $a_{ij}$ and $\psi^i$ in favor of $\Ione$, $\Itwo$. Using  $g\geq 1$ and the definitions of $\psi^i, \Ione,\Itwo$ {and $[[\cdot]]_{C_b^2(\R^{N})}$} {(ref. (\ref{def:psii},\ref{Ionefinite}, \ref{Itwofinite}) and \eqref{defseminorm})}
		we reach the following expression:
		\begin{align*}
		(\cL_S-\partial_s)f_s^1  &
		\leq CN \Ione [[\cPbar_sf]]_{C_b^2(\R^{N})}+CN \Itwo [[\cPbar_sf]]_{C_b^2(\R^{N})}+ C \sum_i \sum_j a_{ij} \psi^j  
		\lvert \partial_i\cPbar_sf(\bfx)\rvert \\
		& +C \sum_i\sum_k \psi^i \lvert \partial_k \partial_i\cPbar_sf(\bfx)\rvert 
		+C \sum_i\sum_k \psi^i \psi^k \lvert \partial_k \partial_i\cPbar_sf(\bfx)\rvert \,.
		\end{align*}
	}
	{
		The last three terms above are dealt with using Cauchy-Schwarz's inequality, and $y^{1/2}\leq (1+y)/2$ for $y\geq 0$:
		\begin{align*}
		\sum_i \sum_j a_{ij} \psi^j  
		\lvert \partial_i\cPbar_sf(\bfx)\rvert & \leq 
		\sum_j \psi^j \left(\sum_i a_{ij}^2\right)^{1/2} [[\cPbar_sf]]_{C_b^2(\R^{N})} \\
		&\leq [[\cPbar_sf]]_{C_b^2(\R^{N})}\frac12 \sum_j \psi^j (1+\psi^j)
		\leq \frac12 N(\Ione+\Itwo) [[\cPbar_sf]]_{C_b^2(\R^{N})}~;\\
		\sum_i \sum_k \psi^i  
		\lvert \partial_i\partial_k \cPbar_sf(\bfx)\rvert & \leq 
		[[\cPbar_sf]]_{C_b^2(\R^{N})} \left(\sum_{i,k}(\psi^i)^2\right)^{1/2}\\
		&\leq N\sqrt{\Itwo} [[\cPbar_sf]]_{C_b^2(\R^{N})} \leq \frac12  N(1+\Itwo)[[\cPbar_sf]]_{C_b^2(\R^{N})}~; \\
		\sum_i \sum_k \psi^i \psi^k 
		\lvert \partial_i\partial_k \cPbar_sf(\bfx)\rvert & \leq 
		[[\cPbar_sf]]_{C_b^2(\R^{N})} \left(\sum_{i,k}(\psi^i)^2(\psi^k)^2\right)^{1/2}\\
		&\leq N \Itwo [[\cPbar_sf]]_{C_b^2(\R^{N})}~.
		\end{align*}
	}
	{We now put everything together, and apply the semigroup $\cP_{t-s}^{N,\varepsilon}$, obtaining
		\begin{align*}
		\cP^{N,\varepsilon}_{t-s}(\cL_S-\partial_s)f_s^1
		&\leq  CN[[\cPbar_sf]]_{C_b^2(\R^N)} \left(1+\cP^{N,\varepsilon}_{t-s}\Ione(\bfx,\bfa) +\cP^{N,\varepsilon}_{t-s}\Itwo(\bfx,\bfa) \right),  
		\end{align*}
	}
	{where $C$ has been redefined.}
	By Lemma \ref{lem:Ione} and Lemma \ref{lem:Itwo} we then get
	\begin{align*}
	\cP^{N,\varepsilon}_{t-s}(\cL_S-\partial_s)f_s^1
	&\leq CN[[\cPbar_sf]]_{C_b^2(\R^N)} {\left(1+2C_1 + (1+C_3)\Ione(\bfx,\bfa)e^{-C_2\frac{t}{\varepsilon}}+ \Itwo(\bfx,\bfa)e^{-C_4\frac{t}{\varepsilon}} \right)}  \\
	&\leq {3}CN[[\cPbar_sf]]_{C_b^2(\R^N)} \left(C_1 + C_3\Ione(\bfx,\bfa)e^{-C_2\frac{t}{\varepsilon}}+ \Itwo(\bfx,\bfa)e^{-C_4\frac{t}{\varepsilon}} \right) .
	\end{align*}
	(The inequality holds provided $C_1\geq 1,C_3\geq 2$, which may be assumed without loss of generality.) Then Step 2 follows from applying Lemma \ref{lem:derivativeest}.
	Finally we prove Step 3. By the definition of $\Ione$ we have
	\begin{equation}
	\lvert\fone (\bfx,\bfa)\rvert\leq CN \Ione(\bfx,\bfa) {[[\cPbar_sf]]_{C_b^2(\R^N)}}.
	\end{equation}
	The inequality \eqref{S3} then follows from Lemma \ref{lem:derivativeest} and the proof is complete. 
\end{proof}

\subsection{Proof of Lemma \ref{lem:Ione}, \ref{lem:Itwo} and \ref{lem:derivativeest}}
We start by proving the bounds on the semigroup $\cPbar_t$ contained in  Lemma \ref{lem:derivativeest} and then prove Lemma \ref{lem:Ione} and \ref{lem:Itwo}.  Throughout this section, without loss of generality, we set the diffusion coefficient $D=1$ and $n=1$, so that $f \in C_b^2(\R^N)$ and $K,V:\R \rightarrow \R$. We will therefore write $\partial_i$ instead of $\nabla_i$ and $V'$ denotes derivative of $V$ with respect to its argument (similarly for $K'$).
\begin{proof}[Proof of Lemma \ref{lem:derivativeest}]
	Define the function
	\begin{equation}\label{defGamma}
	\Gamma(f):=\alpha\left\vert f(\bfx) \right\vert^2+\beta\sum_{i=1}^N\left\vert \partial_i f(\bfx)\right\vert^2+\sum_{i,j=1}^N \left\vert \partial_{ij}^2f(\bfx) \right\vert^2 \, 
	\end{equation}
	where $\alpha, \beta$ are positive constants to be chosen later. 
	Note that if we can prove
	\begin{equation}\label{eq:boundedderivatives}
	\Gamma(\cPbar_t f) \leq C \left(\sum_{i=1}^N\lVert \partial_{i}f \rVert_\infty^2+ \sum_{i,j=1}^N\lVert \partial_{ij}^2f \rVert_\infty^2\right)
	\end{equation}
	for some constant $C$ dependent only on $\kappa_1,\kappa_2,\kappa_3$, and $\kappa_4$, then \eqref{eq:C2toC2ineq} holds. On the other hand, to prove \eqref{eq:expdecayofderivative}  it is  sufficient to take $\alpha=0$ in the definition of $\Gamma$ and, with this choice of $\alpha$, show that for some $\delta>0$ one has
	\begin{equation}\label{eq:gronwallrequirement} 
	(\partial_s-\bar\cL^N)\Gamma(f_s)\leq -\delta \Gamma(f_s) \, ,
	\end{equation}
	where we are using the shorthand  notation  $f_t(\bfx):=(\bar\cP_t^Nf )\bfx$. 
	Indeed, by standard semigroup properties,  $\partial_s(\cPbar_{t-s}\Gamma(f_s)) =\cPbar_{t-s}((\partial_s-\bar\cL^N)\Gamma(f_s))$, and therefore \eqref{eq:gronwallrequirement} implies
	\begin{equation*}
	\partial_s(\cPbar_{t-s}\Gamma(f_s))  \leq -\delta \cPbar_{t-s}\Gamma(f_s).
	\end{equation*}
	If we apply  Gronwall's inequality (in the $s$ variable) to the above and calculate the resulting expression in $s=t$,  we obtain
	\eqh{\Gamma (\cPbar_t f)\leq e^{-\delta t}\cPbar_t(\Gamma(f))\, ;} 
	{by Markovianity, $\cPbar_t$ is a contraction in the $L^\infty$ norm}, thus \eqref{eq:expdecayofderivative} holds. 
	
	Similarly, in order to prove   \eqref{eq:boundedderivatives}, one just needs to show that \eqref{eq:gronwallrequirement} holds with  $\delta=0$, however this time with $\alpha>0$ to be chosen, as we will show.
	To this end, we expand the expression  $(\partial_s-\bar\cL^N)\Gamma(f_s)$ as follows:
	\begin{align*}
	(\partial_s-\bar\cL^N)\Gamma(f_s) =& -2\alpha\sum_{i=1}^N\lvert \partial_i f_s\rvert^2+2\beta\sum_{i=1}^N( [\partial_i,\bar\cL^N]f_s)(\partial_if_s) \\
	-& 2\beta\sum_{i,j=1}^N\lvert \partial_{ij}^2f_s\rvert^2 + 2\sum_{i,j=1}^N( [\partial_{ij}^2,\bar\cL^N]f_s)(\partial_{ij}^2f_s) - 2\sum_{i,j,k=1}^N\lvert \partial_{ijk}^3f_s\rvert^2\, ,
	\end{align*}
	where in the above $[\cdot, \cdot]$ denotes the commutator\footnote{We recall that if $X$ and $Y$ are two differential operators, the commutator $[X,Y]$ is the differential operator defined as $[X,Y]f):= (XY)(f)-(YX)(f)$. } between two differential operators. The above equality can be obtained via direct calculations or by using shortcuts similar to those typically employed in Bakry-Emery-type approaches, see for example \cite[equation (26)]{CrisanOttobre} or also \cite{zeg}.  
	Using  the definition of the function $\adr^i$, see \eqref{fbari}, and  \eqref{10star}, we can write
	\begin{equation*}
	\bar{\cL}^Nf(\bfx) = \sum_{i=1}^N \adr^i(\bfx)\partial_if(\bfx) + \sum_{i=1}^N \partial_{ii}^2f(\bfx).
	\end{equation*}
	Now
	\begin{align*}
	[\partial_i,\bar\cL^N] &= \sum_{j=1}^N[\partial_i, \adr^j(\bfx)\partial_j]=\sum_{j=1}^N(\partial_i \adr^j(\bfx))\partial_j,
	\end{align*}
	and similarly
	\begin{align*}
	[\partial^2_{ij},\bar\cL^N]&= \sum_{k=1}^N[\partial^2_{ij}, \adr^k(x)\partial_k]\\
	&=\sum_{k=1}^N (\partial_{ij}^2 \adr^k(\bfx))\partial_k+(\partial_i\adr^k(\bfx))\partial_{jk}^2 + (\partial_j\adr^k(\bfx))\partial_{ik}^2.
	\end{align*}
	
	For $i,j,k$ all distinct we have
	\begin{align*}
	\partial_{i}\adr^i(\bfx)&=\sum_{\ell\neq i} \partial_i\Kbar(x^\ell-x^i) -\partial^2_{ii} V(x^i)\\
	&= -\sum_{\ell\neq i} (\Kbar)'(x^\ell-x^i) - V''(x^i) \leq \kappa_1-\kappa_3,\\
	\partial_{i}\adr^j(\bfx) &= {(\Kbar)'(x^i-x^j) \leq \kappa_1/N,}
	\end{align*}
	and 
	\begin{align*}
	\partial_{ij}\adr^k(\bfx) &=0\\
	\lvert\partial_{ii}^2\adr^k(\bfx)\rvert &= \lvert (\Kbar)''(x^i-x^k)\rvert \leq \kappa_2/N\\
	\lvert\partial_{ik}^2\adr^k(\bfx)\rvert &=\lvert-(\Kbar)''(x^i-x^k)\rvert\leq \kappa_2/N\\
	\lvert\partial_{ii}^2\adr^i(\bfx)\rvert &=\left| \sum_{\ell\neq i}  (\Kbar)''(x^\ell - x^i)-V'''(x^i)\right| \leq \kappa_2+\kappa_4.
	\end{align*}
	Therefore we have
	\begin{align*}
	(\partial_s-\bar\cL^N)\Gamma(f_s) =& -2\alpha\sum_{i=1}^N\lvert \partial_i f_s\rvert^2+2\beta\sum_{i=1}^N \sum_{j=1}^N(\partial_i \adr^j(\bfx))(\partial_j f_s)(\partial_if_s) \\
	&- 2\beta\sum_{i,j=1}^N\lvert \partial_{ij}^2f_s\rvert^2 + 2\sum_{i,j,k=1}^N (\partial_{ij}^2 \adr^k(\bfx))(\partial_kf_s)(\partial_{ij}^2f_s)\\
	&+2\sum_{i,j,k=1}^N(\partial_i\adr^k(\bfx))(\partial_{jk}^2f_s)(\partial_{ij}^2f_s) - 2\sum_{i,j,k=1}^N\lvert \partial_{ijk}^3f_s\rvert^2\\
	\leq& -2\alpha\sum_{i=1}^N\lvert \partial_i f_s\rvert^2+2\beta\sum_{i=1}^N (\kappa_1-\kappa_3)(\partial_i f_s)^2\\
	&+\frac{2\beta\kappa_1}{N}\sum_{i=1}^N \sum_{j\neq i}\lvert\partial_j f_s\rvert\lvert\partial_if_s\rvert - 2\beta\sum_{i,j=1}^N\lvert \partial_{ij}^2f_s\rvert^2 \\
	&+ 2\sum_{i}^N (\kappa_2+\kappa_4)\lvert\partial_if_s\rvert\lvert\partial_{ii}^2f_s\lvert+ \frac{4\kappa_2}{N}\sum_{i}\sum_{j\neq i} \lvert\partial_if_s\rvert\lvert\partial_{ij}^2f_s\rvert\\
	&+ \frac{2\kappa_2}{N}\sum_{i}^N\sum_{k\neq i} \lvert\partial_kf_s\rvert\lvert\partial_{ii}^2f_s\rvert\\
	&+\frac{2\kappa_1}{N}\sum_{i,j=1}^N\sum_{k\neq i}\lvert\partial_{jk}^2f_s\rvert\lvert\partial_{ij}^2f_s\rvert+2(\kappa_1-\kappa_3)\sum_{i,j}^N(\partial_{ij}^2f_s)^2.
	\end{align*}
	Now by using Young's inequality we have
	\begin{align*}
	&(\partial_s-\bar\cL^N)\Gamma(f_s)     \leq -2\alpha\sum_{i=1}^N\lvert \partial_i f_s\rvert^2+2\beta\sum_{i=1}^N (\kappa_1-\kappa_3)(\partial_i f_s)^2+\beta\kappa_1\sum_{i=1}^N \lvert\partial_if_s\rvert^2 \\
	&+\beta\kappa_1\sum_{j=1}^N\lvert\partial_j f_s\rvert^2- 2\beta\sum_{i,j=1}^N\lvert \partial_{ij}^2f_s\rvert^2 + (\kappa_2+\kappa_4)\sum_{i=1}^N \lvert\partial_if_s\rvert^2+(\kappa_2+\kappa_4)\sum_{i=1}^N \lvert\partial_{ii}^2f_s\lvert^2\\
	&+ \frac{2\kappa_2}{N}\sum_{i}\sum_{j\neq i} \lvert\partial_{ij}^2f_s\rvert^2+2\kappa_2\sum_{i} \lvert\partial_if_s\rvert^2+ \kappa_2\sum_{k=1}^N \lvert\partial_kf_s\rvert^2+\kappa_2\sum_{i}^N \lvert\partial_{ii}^2f_s\rvert^2\\
	&+\kappa_1\sum_{j=1}^N\sum_{k\neq i}\lvert\partial_{jk}^2f_s\rvert^2+\kappa_1\sum_{i,j=1}^N\lvert\partial_{ij}^2f_s\rvert^2+2(\kappa_1-\kappa_3)\sum_{i,j}^N(\partial_{ij}^2f_s)^2.
	\end{align*}
	Finally, collecting all terms gives
	\begin{align*}
	(\partial_s-\bar\cL^N)\Gamma(f_s)     \leq& \left(-2\alpha-2\beta\kappa_3+4\beta\kappa_1 + \kappa_4+4\kappa_2\right)\sum_{i=1}^N\lvert \partial_i f_s\rvert^2\\
	&+\left(- 2\beta+2\kappa_2+\kappa_4+ \frac{2\kappa_2}{N}+4\kappa_1-2\kappa_3\right)\sum_{i,j}^N(\partial_{ij}f_s)^2.
	\end{align*}
	If we choose
	$$2\beta\geq \max \{0,2\kappa_2+\kappa_4+ \frac{2\kappa_2}{N}+4\kappa_1-2\kappa_3\}$$
	and 
	$${2\alpha=4\beta\kappa_1 + \kappa_4+4\kappa_2 \,}$$ 
	then we have $(\partial_s-\bar\cL)\Gamma(f_s)\leq 0$ as required to prove \eqref{eq:boundedderivatives}. 
	It should be emphasized that $\alpha$ and $\beta$ can be taken independent of $N$. Moreover if in the above calculations we set $\alpha=0$ and choose $\delta>0$ and $\beta>0$ such that $$2\beta=\delta +4\kappa_2+\kappa_4+4\kappa_1,$$ 
	and
	\begin{equation*}
	{  -\kappa_3 \leq -\frac{\delta}{2}-2\kappa_1-\frac{\kappa_4}{2\beta}-\frac{2\kappa_2}{\beta},}
	\end{equation*}
	we have
	\begin{align*}
	(\partial_s-\bar\cL^N)\Gamma(f_s)     &\leq -\delta\lr{\beta \sum_{i=1}^N\lvert \partial_i f_s\rvert^2+\sum_{i,j}^N(\partial_{ij}f_s)^2}=
	-\delta \Gamma(f_s),
	\end{align*}
	where $\Gamma(f_s)$ is given by \eqref{defGamma} with $\alpha=0$. This completes the proof of \eqref{eq:expdecayofderivative}.
\end{proof}

\begin{proof}[Proof of Lemma \ref{lem:Ione}]\label{sec:proofdetails_lemma}
	As in previous proofs, without loss of generality we set the diffusion coefficient $D=1$ and $n=1$.
	From \eqref{def:Lf} and \eqref{def:psii},
	\begin{align*}
	\cL_F\psi^i(\bfx,\bfa) &= \sum_{j:j\neq i}\left(-\nu_da_{ij}+\nu_f(x^i,x^j)(1-a_{ij})\right) g(x^j-x^i)\\
	&\leq -\nu_d \psi^i(\bfx,\bfa)+\sum_{j:j\neq i}\nu_f(x^i,x^j)g(x^j-x^i)\, ,
	\end{align*}
	having used the fact that $1-a_{ij}\leq 1$. 
	Note that if $\lvert x^j-x^i\rvert>R$ then $\tilde{\nu}_f(x^i,x^j)=\nu_f^N\varphi_R(|x^j-x^i|)=0$,  hence
	\begin{equation}\label{3dots}
	\tilde{\nu}_f(x^i,x^j)g(x^j-x^i) \leq \nu_f^N\sup_{\lvert z\rvert \leq R}g(z)\leq \frac{C}{N}, 
	\end{equation}
	where the constant on the r.h.s. depends on $R$, 
	so that
	$$\sum_{j:j\neq i}\tilde{\nu}_f(x^i,x^j)g(x^j-x^i) \leq N\nu_f^N\sup_{\lvert z\rvert \leq R}g(z)\leq C$$ for some constant $C$ which depends on $R$ but  is independent of $N$.  In what follows $C$ may vary from  line to line. We then have
	\begin{equation*}
	\cL_F\psi^i(\bfx,\bfa) \leq -\nu_d \psi^i(\bfx,\bfa)+C.
	\end{equation*}
	Taking the sum over $i$ and dividing by $N$ on both sides we obtain
	\begin{equation}\label{bla}
	\cL_F\Ione(\bfx,\bfa) \leq -\nu_d \Ione(\bfx,\bfa)+C.
	\end{equation}
	For $j \neq i$, using \eqref{def:Ls} we get
	\begin{align*}
	\cL_S(g(x^j-x^i))& = \sum_{k=1}^N\sum_{r\neq k}K(x^r-x^k)a_{rk}\partial_{k}(g(x^j-x^i)) +\sum_{k=1}^N\partial_{kk}^2(g(x^j-x^i))\\
	&- \sum_{k=1}^N \partial_kV(x^k)\partial_k(g(x^j-x^i))\\
	&=\sum_{r\neq j} K(x^r-x^j)a_{rj}g'(x^j-x^i){-}\sum_{r\neq i} K(x^r-x^i)a_{ri}g'(x^j-x^i)\\
	&+ 2g''(x^j-x^i)
	+V'(x^i)g'(x^j-x^i)-  V'(x^j)g'(x^j-x^i)\\
	&\leq C\sum_{r\neq j} g(x^r-x^j)a_{rj}+C\sum_{r\neq i} g(x^r-x^i)a_{ri}+ 2C+C\lvert x^i-x^j\rvert \,.
	\end{align*}
	To obtain the last inequality we have used that $g'$ is bounded, $V'$ is Lipschitz and  $\lvert K(z)\rvert \leq g(z)$. Now using  $\lvert z\rvert \leq g(z)$ and the definition of $\psi^i$  \eqref{def:psii}, we have
	\begin{align}\label{eq:LSg}
	\cL_S(g(x^j-x^i)) \leq C\psi^j(\bfx,\bfa)+C\psi^i(\bfx,\bfa)+ 2C+Cg( x^j-x^i),
	\end{align}	
	and so
	\begin{align*}
	\cL_S\psi^i(\bfx,\bfa) &= \sum_{j:j\neq i}a_{ij} \cL_S(g(x^j-x^i))\\
	&\leq \sum_{j:j\neq i}a_{ij} C\psi^j(\bfx,\bfa)+CN\psi^i(\bfx,\bfa)+ 2C\sum_{j:j\neq i}a_{ij}+\sum_{j:j\neq i}a_{ij} Cg( x^i-x^j)\\
	&\leq C\sum_{j:j\neq i} \psi^j(\bfx,\bfa)+CN\psi^i(\bfx,\bfa)+ 2C\psi^i(\bfx,\bfa)+C\psi^i(\bfx,\bfa).
	\end{align*}
	To obtain the last inequality we have used  $1\leq g(x^i-x^j)$ and once again the definition of $\psi^i$. Taking the sum over $i$ and dividing by $N$ we finally obtain
	\begin{align}\label{blabla}
	\cL_S\Ione(\bfx,\bfa) \leq 2CN\Ione(\bfx,\bfa)+ 2C\Ione(\bfx,\bfa)+C\Ione(\bfx,\bfa) \leq CN\Ione(\bfx,\bfa).
	\end{align}	
	We can now estimate $\cP_t^{N,\ep} \Ione(\bfx,\bfa)$; indeed, using the fact that the semigroup and its generator commute, together with \eqref{bla} and \eqref{blabla}, 
	\begin{align*}
	\partial_t\cP_t^{N,\ep}\Ione(\bfx,\bfa)&= \cP_t^{N,\ep}\left(\frac{1}{\varepsilon}\cL_F\Ione+\cL_S\Ione\right)(\bfx,\bfa)\\
	&\leq -\frac{\nu_d}{\varepsilon}\cP_t^{N,\ep}\Ione(\bfx,\bfa)+\frac{C}{\varepsilon} +CN\cP_t^{N,\ep}\Ione(\bfx,\bfa).
	\end{align*}
	Now, by Gronwall's inequality, we obtain
	\begin{equation*}
	\cP_t^{N,\ep}\Ione(\bfx,\bfa) \leq \frac{C}{\nu_d-CN\varepsilon} + \Ione(\bfx,\bfa) e^{\left(CN\varepsilon-\nu_d\right)\frac{t}{\varepsilon}}.
	\end{equation*}
	Therefore, for $N\varepsilon$ sufficiently small we have
	\begin{equation}\label{eq:pisiineq}
	\cP_t^{N,\ep}\Ione(\bfx,\bfa) \leq C_1 + \Ione(\bfx,\bfa) e^{-C_2\frac{t}{\varepsilon}}
	\end{equation}
	where $C_1,C_2$ are positive constants independent of $N$ and $\varepsilon$ (as long as $N\ep \ll 1$). 
\end{proof}

\begin{proof}[Proof of Lemma \ref{lem:Itwo}] 
	The proof is similar in spirit to the proof of Lemma \ref{lem:Ione}. Note that
	\begin{align*}
	\psi^i(\bfx,\bfa)^2 &= \sum_{j:j\neq i}\sum_{k:k\neq i} a_{ij}a_{ik}g(x^j-x^i)g(x^k-x^i)\\
	&=\sum_{j:j\neq i}\sum_{k:k\neq i,j} a_{ij}a_{ik}g(x^j-x^i)g(x^k-x^i) +\sum_{j:j\neq i} a_{ij}g(x^j-x^i)^2.
	\end{align*}
	When we apply $\cL_F$ to $\psi^i(\bfx,\bfa)^2$ we get
	\begin{align*}
	\cL_F(\psi^i(\bfx,\bfa)^2) &= -2\nu_d\sum_{j:j\neq i}\sum_{k:k\neq i,j} \!\!\!\!a_{ij}a_{ik}g(x^j-x^i)g(x^k-x^i) \!-\!\nu_d \sum_{j:j\neq i} a_{ij}g(x^j-x^i)^2\\
	&+\sum_{j:j\neq i}\sum_{k:k\neq i,j} \tilde\nu_f(x^i,x^j)(1-a_{ij})a_{ik}g(x^j-x^i)g(x^k-x^i) \\
	&+\sum_{j:j\neq i}\sum_{k:k\neq i,j}\tilde\nu_f(x^i,x^k) a_{ij}(1-a_{ik})g(x^j-x^i)g(x^k-x^i)\\
	&+\sum_{j:j\neq i} \tilde\nu_f(x^i,x^j)(1-a_{ij})g^2(x^j-x^i)\\
	&\leq -\nu_d\psi^i(\bfx,\bfa)^2+\sum_{j:j\neq i}\sum_{k:k\neq i,j} \nu_f(x^i,x^j)a_{ik}g(x^j-x^i)g(x^k-x^i) \\
	&+\!\!\!\sum_{j:j\neq i}\sum_{k:k\neq i,j}\!\!\!\!\tilde\nu_f(x^i,x^k) a_{ij}g(x^j-x^i)g(x^k-x^i)\!+\!\!\sum_{j:j\neq i}\!\! \tilde\nu_f(x^i,x^j)g^2(x^j-x^i).
	\end{align*}
	By using the same reasoning as in \eqref{3dots} (with $g^2$ instead of $g$),
	\begin{align*}
	\cL_F(\psi^i(\bfx,\bfa)^2) &\leq -\nu_d\psi^i(\bfx,\bfa)^2+\sum_{k:k\neq i,j} a_{ik}g(x^k-x^i)+\sum_{j:j\neq i} a_{ij}g(x^j-x^i)+C\\
	&\leq -\nu_d\psi^i(\bfx,\bfa)^2+2\psi^i(\bfx,\bfa)+C \, ,
	\end{align*}
	which gives
	\begin{align}\label{eq:LFItwo}
	\cL_F\Itwo(\bfx,\bfa) \leq -\nu_d\Itwo(\bfx,\bfa)+\frac{2}{N}\sum_{i=1}^N\psi^i(\bfx,\bfa)+C
	\end{align}
	Now we shall calculate $\cL_S\Itwo$. Using \eqref{eq:LSg} we have
	\begin{align*}
	\cL_S(g(x^j-x^i)g(x^k-x^i)) &=\cL_S[g(x^j-x^i)]g(x^k-x^i)+ g(x^j-x^i)\cL_S[g(x^k-x^i)]\\
	&+\sum_{\ell=1}^N\partial_\ell(g(x^j-x^i))\partial_\ell(g(x^k-x^i))\\
	&\leq C\psi^j(\bfx,\bfa)g(x^k-x^i)
	+C\psi^i(\bfx,\bfa)g(x^k-x^i)\\
	&+ 2Cg(x^k-x^i)+2Cg( x^j-x^i)g(x^k-x^i) \\
	&+ C\psi^k(\bfx,\bfa)g(x^j-x^i)+C\psi^i(\bfx,\bfa)g(x^j-x^i)\\
	& + 2Cg(x^j-x^i)+g'(x^j-x^i)g'(x^k-x^i).
	\end{align*}
	Now we can insert this estimate into the definition of $\Itwo$ and, using the fact that $\cLs$ only  acts in the variable $\bfx$, 
	\begin{align*}
	\!\!\cL_S(\psi^i(\bfx,\bfa)^2)\!&\leq\! C\!\!\sum_{j:j\neq i}\sum_{k:k\neq i}\!\! a_{ij}a_{ik}\psi^j(\bfx,\bfa)g(\!x^k\!-\!x^i)\!+\!C\!\sum_{j:j\neq i}\sum_{k:k\neq i}\!\!\! a_{ij}a_{ik}\psi^i(\bfx,\bfa)g(\!x^k\!-\!x^i)\\
	&+ \!\!2C\sum_{j:j\neq i}\sum_{k:k\neq i}\!\! a_{ij}a_{ik}g(x^k-x^i)+2C\!\!\sum_{j:j\neq i}\sum_{k:k\neq i}\!\! a_{ij}a_{ik}g( x^j-x^i)g(x^k-x^i) \\
	&+ \!\!C\!\!\!\sum_{j:j\neq i}\sum_{k:k\neq i} \!\!\!a_{ij}a_{ik}\psi^k(\bfx,\bfa)g(x^j-x^i)\!+\!C\!\sum_{j:j\neq i}\sum_{k:k\neq i} \!\!\!a_{ij}a_{ik}\psi^i(\bfx,\bfa)g(x^j-x^i)\\
	&+\!\! 2C\sum_{j:j\neq i}\sum_{k:k\neq i} a_{ij}a_{ik}g(x^j-x^i)+\sum_{j:j\neq i}\sum_{k:k\neq i} a_{ij}a_{ik}g'(x^j-x^i)g'(x^k-x^i)\,.
	\end{align*}
	Using that $g'$ is bounded and $a_{ij}\leq 1$, we obtain
	\begin{align*}
	\cL_S(\psi^i(\bfx,\bfa)^2)&\leq C\sum_{j:j\neq i} \psi^j(\bfx,\bfa)\psi^i(\bfx,\bfa)+CN\psi^i(\bfx,\bfa)^2\\
	&+ 2C\sum_{j:j\neq i} a_{ij}\psi^i(\bfx,\bfa)+2C\psi^i(\bfx,\bfa)^2\\
	&+ C\sum_{k:k\neq i} \psi^k(\bfx,\bfa)\psi^i(\bfx,\bfa)+CN \psi^i(\bfx,\bfa)^2\\
	&+ 2C\sum_{j:j\neq i}\sum_{k:k\neq i} a_{ij}a_{ik}g(x^j-x^i)+C\sum_{j:j\neq i}\sum_{k:k\neq i} a_{ij}a_{ik}.
	\end{align*}
	Since $1\leq g$ we can write each of these expressions in terms of $\psi^i, \psi^j$.
	\begin{align*}
	&\cL_S(\psi^i(\bfx,\bfa)^2)\leq 2C\sum_{j:j\neq i} \psi^j(\bfx,\bfa)\psi^i(\bfx,\bfa)+2CN\psi^i(\bfx,\bfa)^2+ 7C\psi^i(\bfx,\bfa)^2.
	\end{align*}
	By Young's inequality
	\begin{align*}
	&\cL_S(\psi^i(\bfx,\bfa)^2)\leq C\sum_{j:j\neq i} \psi^j(\bfx,\bfa)^2+CN\psi^i(\bfx,\bfa)^2+2CN\psi^i(\bfx,\bfa)^2+ 7C\psi^i(x,a)^2.
	\end{align*}
	Summing over $i$, 
	\begin{align}
	\cL_S\Itwo(\bfx,\bfa) &\leq CN\Itwo(\bfx,\bfa)+CN\Itwo(\bfx,\bfa)+2CN\Itwo(\bfx,\bfa)+ 7C\Itwo(\bfx,\bfa) \nonumber\\
	&\leq CN\Itwo(\bfx,\bfa).\label{eq:LSItwo}
	\end{align}
	Finally, combining \eqref{eq:LFItwo} and \eqref{eq:LSItwo},  we have
	\begin{align*}
	\partial_t\cP_t^{N,\ep} \Itwo(\bfx,\bfa) &= \cP_t^{N,\ep}(\frac{1}{\varepsilon}\cL_F\Itwo(\bfx,\bfa) + \cLs\Itwo(\bfx,\bfa))\\
	&\leq \cP_t^{N,\ep}\left(-\frac{\nu_d}{\varepsilon}\Itwo(\bfx,\bfa) +\frac{2}{\varepsilon N}\sum_{i=1}^N\psi^i(\bfx,\bfa)+\frac{C}{\varepsilon}+ CN\Itwo(\bfx,\bfa)\right).
	\end{align*}
	By \eqref{eq:pisiineq},
	\begin{align*}
	\partial_t\cP_t^{N,\ep} \Itwo(\bfx,\bfa) &\leq \left(-\frac{\nu_d}{\varepsilon}+CN\right)\cP_t^{N,\ep}\Itwo(\bfx,\bfa) +\frac{C}{\varepsilon}+\frac{2}{\varepsilon}C_1+\frac{2}{\varepsilon}\Ione(\bfx,\bfa)e^{-C_2\frac{t}{\varepsilon}}.
	\end{align*}
	If there exists $\varepsilon_0>0$ such that  $N\varepsilon< \varepsilon_0$ then we can integrate this inequality to find constants $C_1,C_2,C_3,C_4>0$ such that
	\begin{equation*}
	\cP_t^{N,\ep}\Itwo(\bfx,\bfa) \leq C_1 + C_3\Ione(\bfx,\bfa)e^{-C_2\frac{t}{\varepsilon}}+ \Itwo(\bfx,\bfa)e^{-C_4\frac{t}{\varepsilon}}.
	\end{equation*}
\end{proof}

\section{The large $N$ limit: proof of \eqref{EEpart} and \eqref{EEpartunif}}\label{Sec:MF}
The purpose of this section is to estimate the term ${E}^{\rm part}$ given by \eqref{Ebar}. {In the original dynamics \eqref{slow}-\eqref{fast}, any particle typically interacts at a given time with a bounded (in $N$) number of other particles; hence mean-field techniques cannot be employed. However, the averaged process \eqref{slow2} is of mean-field type, so that classical techniques can be used.}
We introduce an auxiliary system of independent particles $Y^i_t$, such that
$Y^{i,N}_0=X^{i,N,\ep}_0=\bar{X}_0^{i,N} $ and
\eq{\label{interY}
	dY^{i,N}_t = - \nabla V(Y_t^{i,N})dt+ \left(\bar{K}\ast \rho_t \right) (Y^{i,N}_t) dt +\sqrt{2D}dB^i_t~.
}
The Brownian motions $B^i_t$ are the same as those for the $\bar{X}^i$ in equation \eqref{slow2}. From now on for simplicity we drop the superscript $N$ and we simply use the notation $Y_t^i, \bar{X}_t^i$, respectively,  instead of {$Y_t^{i,N},\bar X_t^{i,N}$}, respectively.  The $Y^i$ are independent, and the law $\mu_t(y)dy$ of each $Y^i$ solves the equation (recall that the $X^i_0$ are iid with law $\rho_0$): 
\eq{
	\partial_t \mu_t(y) &=\nabla\cdot(\mu_t\nabla V(y))-\nabla \cdot \left( (\bar{K}\ast\rho_t)\mu_t \right)+ D\Delta \mu_t ,\\
	\mu_t(y)|_{t=0}&=\rho_0(y).
}
Note that, under Assumption \ref{H3new} and Hypothesis \ref{H1}, equation \eqref{macro} has a unique solution (say in $L^2$)\footnote{While hard to find in the literature, this is standard under our assumptions on $\rho_0$, $V$, $K$; indeed much more general results are available for this equation, e.g. uniqueness of measure-valued solutions, but this is not needed here.} So, comparing \eqref{macro} and the above equation, we conclude that the law of any $Y_i$ at time $t$ is precisely $\rho_t(y)dy$. 
We therefore get
\begin{align}\label{EEbar}
E^{\rm part} & = \mathbb{E}\lr{\frac1N\sum_{i}u(\bar X^i_t)-\int u(y) \rho_t(y) dy}^2 \nonumber\\
&\leq  2 \mathbb{E}\lr{\frac1N\sum_{i}u(\bar X^i_t)-\frac1N\sum_{i}u(Y^i_t)}^2 \nonumber\\
&+2\mathbb{E}\lr{\frac1N\sum_{i}u(Y^i_t)-\int u(y) \rho_t(y) dy}^2 \nonumber \\
& \leq 2L_u^2\mathbb{E}\left[(\bar{X}^i_t-Y^i_t)^2\right]
\nonumber\\
& +2\mathbb{E}\lr{\frac1N\sum_{i}u(Y^i_t)-\int u(y) \rho_t(y) dy}^2,
\end{align}
where $L_u$ stands for the Lipshitz constant of $u$ and we have used the particles'  exchangeability. Our goal now is to estimate both terms in \eqref{EEbar}. The main tool is the following classical lemma.
\begin{lemma}\label{lem:step2}
	Assume Hypotheses \ref{H1} and Assumption \ref{H3new}. Then there exists a constant $C$ such that for any time $T>0$
	\[
	\sup_{t\in[0, T]}\mathbb{E}\left[\left(\bar{X}_t^i-Y_t^i\right)^2
	\right] \leq \frac{C}{N}e^{4LT},
	\]
	where $L$ is related to the Lipschitz constants of $\bar{K}$ and $\nabla V$.\\
	Assume in addition that Hypothesis \ref{H2} holds with $\kappa_3>\kappa_{\rm mf}:=2L_{\bar{K}}$, where $L_{\bar{K}}$ is the Lipschitz constant of $\bar{K}$. Then there exists a constant $C$ such that
	\[
	\sup_{t\geq 0}\mathbb{E}\left[\left(\bar{X}_t^i-Y_t^i\right)^2\right]
	\leq \frac{C}{N}.
	\]
\end{lemma}
\begin{proof}[Proof of Lemma \ref{lem:step2}]
	We first compare equations \eqref{slow2} with \eqref{interY}, to obtain
	\eq{
		d\lr{\bar{X}^i_ t -Y^i_t}= &- \lr{\nabla V(\bar X_t^i)-\nabla V(Y_t^i)}dt\\
		&+ \left\{\left[\sum_{j\neq i} \bar{K}^N(\bar{X}^j_t-\bar{X}^i_t)\right]-\left(\bar{K}\ast \rho_t \right) (Y^i_t)\right\} dt,
	}
	which is a purely deterministic equation. Therefore, multiplying both sides by \\$\lr{\bar{X}^i_ t \!-\!Y^i_t}$ and integrating with respect to time, we get:
	\eq{\label{firsteq_nosum}
		\frac12  \left(\bar{X}^i_t -Y^i_t\right)^2(t) &= 
		-\int_0^t \left[\nabla V(\bar{X}^i_s)-\nabla V (Y^i_s) \right][\bar{X}^i_s -Y^i_s] ds  \\
		&+\int_0^t\left[\sum_{j\neq i} \lr{\bar{K}^{N} -\frac1N\bar{K}} (\bar{X}^j_s-\bar{X}^i_s)\right][\bar{X}^i_s -Y^i_s] ds \\
		&+\int_0^t\frac1N \left[\sum_{j\neq i} \bar{K} (\bar{X}^j_s-\bar{X}^i_s)-\bar{K} (Y^j_s-Y^i_s)\right][\bar{X}^i_s -Y^i_s] ds \\
		& +\int_0^t \left\{\frac1N \left[\sum_{i\neq j}\bar{K}(Y^j_s-Y^i_s) \right]- (\bar{K}\ast\rho_s)(Y^i_s)\right\}[\bar{X}^i_s-Y^i_s]ds \\ 
		& =: \int_0^t\lr{I_1+\ldots +I_4} ds.
	}
	To estimate $I_1$ we either use the boundedness of the second derivative of $V$ to get
	\[
	\mathbb{E}[I_1] =- \mathbb{E}\left[\nabla V(\bar{X}^i_s)-\nabla V (Y^i_s) \right][\bar{X}^i_s -Y^i_s] \leq C_V\mathbb{E}\left[\left(\bar{X}^i_s -Y^i_s\right)^2\right],
	\]
	or the convexity of the external potential $V$ from Hypothesis \ref{H2} to get:
	\[
	\mathbb{E}[I_1] =- \mathbb{E}\left[\nabla V(\bar{X}^i_s)-\nabla V (Y^i_s) \right][\bar{X}^i_s -Y^i_s] \leq -\kappa_3 \mathbb{E}\left[\left(\bar{X}^i_s -Y^i_s\right)^2\right] ,
	\]
	Using expressions for $\bar K^N$ and $\bar K$, \eqref{eq:Kbardef} and  \eqref{barK}, respectively, and the boundedness of $\bar K$ and $\varphi_R$, we arrive at the inequality
	\[
	|N\bar{K}^N(x)-\bar{K}(x)| \leq  \frac{C}{N}.
	\]
	Therefore, using the particles' exchangeability we obtain 
	\eq{\mathbb{E}[I_2]\leq \frac{C}{N}\mathbb{E}|\bar{X}^i_s -Y^i_s|\leq \frac{C}{N}\lr{\mathbb{E}\left[\left(\bar{X}^i_s -Y^i_s\right)^2\right]}^{1/2}
		.}
	For the term $I_3$ we use the Lipshitz continuity of $\bar K$
	\eqh{
		|I_3| &\leq L_{\bar{K}} \frac1N \sum_{j\neq i} \left[|\bar{X}^j_s-\bar{X}^i_s -(Y^j_s-Y^i_s)||\bar{X}^i_s -Y^i_s|\right]  \\
		&\leq  L_{\bar{K}} \frac1N \sum_{j\neq i} \left[|\bar{X}^j_s-Y^j_s||\bar{X}^i_s -Y^i_s|+|\bar{X}^i_s -Y^i_s|^2\right] 
	}
	and so, taking the expectation and using exchangeability again, we have
	\eq{
		\mathbb{E}[I_3]\leq 2L_{\bar{K}}  \mathbb{E}\left[\left(\bar{X}^i_s -Y^i_s\right)^2\right].
	}
	Finally, for $I_4$,  by direct expansion we obtain
	\begin{equation}
	\begin{split}
	|\mathbb{E}[I_4]|&=\left|\mathbb{E}\left\{\frac1N \left[\sum_{i\neq j}\bar{K}(Y^j_s-Y^i_s) \right]- (\bar{K}\ast\rho_s)(Y^i_s)\right\}[\bar{X}^i_s-Y^i_s]\right|
	\\
	&\leq  \left(\mathbb{E}\left\{\left[\frac1N \sum_{j\neq i}\bar{K}(Y^j_s-Y^i_s)\right] - (\bar{K}\ast\rho_s)(Y^i_s)\right\}^2\right)^{1/2}\!\!\! \lr{\mathbb{E}\left[\left(\bar{X}^i_s -Y^i_s\right)^2\right]}^{1/2}\\
	&\leq \frac{C}{\sqrt{N}}\lr{\mathbb{E}\left[\left(\bar{X}^i_s -Y^i_s\right)^2\right]}^{1/2}+\frac{C'}{N},
	\end{split}
	\end{equation}
	where the pre-factor $C/\sqrt{N}$ and the last addend $C'/N$ in the last inequality follow from a direct computation;  indeed:
	\eq{\label{N12bn}
		&\mathbb{E}\left\{\left[\frac1N \sum_{j\neq i}\bar{K}(Y^j_s-Y^i_s)\right] - (\bar{K}\ast\rho_s)(Y^i_s)\right\}^2\\
		&=\frac1{N^2}\sum_{k,j\neq i}\mathbb{E}\left[\left(\bar{K}(Y_s^j-Y_s^i)-(\bar{K}\ast \rho_s)(Y^i_s)\right)\left(\bar{K}(Y_s^k-Y_s^i)-(\bar{K}\ast \rho_s)(Y^i_s)\right) \right]\\
		&+\frac1{N^2}\sum_{j\neq i}\mathbb{E}\left[\left(\bar{K}(Y_s^j-Y_s^i)-(\bar{K}\ast \rho_s)(Y^i_s)\right)^2 \right] 
		+\left(\frac{1}{N}-\frac{1}{N-1}\right)
		\mathbb \sum_{j\neq i} (\bar{K}\ast \rho_s)(Y^i_s); 
	}
	the last term above is smaller than $\lVert\bar{K}\rVert_\infty /N=C'/N$. Notice now that the random variables $\bar{K}(Y_s^j-Y_s^i)-(\bar{K}\ast \rho_s)(Y^i_s)$, $\bar{K}(Y_s^k-Y_s^i)-(\bar{K}\ast \rho_s)(Y^i_s)$ are independent conditionally on $Y^i_s$ for $k\neq j$, and have zero conditional expectation (because the law of any $Y^l_s$ is equal to the same $\rho_s(y)dy$); thus only the $N-1$ terms in the first term of the second line of \eqref{N12bn} are non zero, and they can be estimated by a constant proportional to $\lVert \bar{K}\rVert^2_\infty$.
	Therefore, taking expectation in \eqref{firsteq_nosum}, and gathering all estimates from above we obtain either
	\eq{
		\label{eq:withH1}
		\frac12\mathbb{E}\left[\left(\bar{X}^i_t -Y^i_t\right)^2\right](t) \leq& \frac{C}{N}\int_0^t \lr{\mathbb{E}\left[\bar{X}^i_s -Y^i_s\right]^2}^{1/2}\, ds +(2L_{\bar{K}} +C_V)\int_0^t\mathbb{E}[ \bar{X}^i_s -Y^i_s]^2\,ds \\
		&+\int_0^t \left[\frac{C}{\sqrt{N}}\lr{\mathbb{E}\left[\bar{X}^i_s-Y^i_s\right]^2}^{1/2}+\frac{C'}{N}\right]ds,
	}
	or, if $V$ satisfies Hypothesis \ref{H2}
	\eq{
		\frac12\mathbb{E}\left[\left(\bar{X}^i_t -Y^i_t\right)^2\right](t) \!\leq& (-\kappa_3+2L_{\bar{K}}) \int_0^t \!\mathbb{E}[\bar{X}^i_s -Y^i_s]^2\, ds
		+\frac{C}{N}\int_0^t\lr{\mathbb{E}\left[\bar{X}^i_s -Y^i_s\right]^2}^{1/2}\, ds\\
		&+\int_0^t\left[\frac{C}{\sqrt{N}}\lr{\mathbb{E}\left[\bar{X}^i_s-Y^i_s\right]^2}^{1/2}+\frac{C'}{N}\right]ds.
	}
	Recall that $\kappa_3>0$, so for $\kappa_3>2L_{\bar{K}}$ we have $-\lambda:=-\kappa_3+2L_{\bar{K}} <0$, and, changing the constant $C$, we can rewrite the above inequality as 
	\begin{equation}
	\alpha'(t)\leq -2\lambda\alpha(t) +\frac{C}{\sqrt{N}}\alpha^{1/2}(t)+\frac{C}{N},
	\label{eq:before_gronwall}
	\end{equation}
	where we denoted $\alpha(t) =\mathbb{E}\left[\left(\bar{X}^i_t -Y^i_t\right)^2\right](t)$. 
	{We choose $\delta>0$ such that $\tilde{\lambda}=\lambda-\delta>0$.
		From \eqref{eq:before_gronwall}, we have
		\begin{equation}
		\alpha'(t)\leq -2(\lambda-\delta)\alpha(t) +\frac{1}{N}\left(\frac{C^2}{8\delta}+C\right) = -2\tilde{\lambda}\alpha(t) +\frac{C}{N},
		\label{eq:before_gronwall2}
		\end{equation}
		where we have redefined $C$ in the last equality. Then Gronwall's lemma ensures
		that $\alpha(t)\leq C/(2\tilde{\lambda}N)$ (recall $\alpha(0)=0$). Hence (again considering  $C$ a generic constant)
		\eq{
			{\sup_{t\geq 0}\mathbb{E}\left[\left(\bar{X}^i_t -Y^i_t\right)^2\right]\leq \frac{C}{N}}.
		}
		In the case when $V$ satisfies only Hypothesis \ref{H1}, we go back to \eqref{eq:withH1}. A similar reasoning as above proves a local-in-time estimate:
		\eq{
			{\mathbb{E}\left[\left(\bar{X}^i_t -Y^i_t\right)^2\right]\leq \frac{C}{N} e^{L t}},
		}
		for some $L>2(2L_{\bar{K}}+C_V)$.
	}
\end{proof}

Now we come back to \eqref{EEbar}. The first term on the r.h.s. is controlled by Lemma \ref{lem:step2}. 
Observe also that the second term on the r.h.s. can be controlled similarly to \eqref{N12bn}, and is therefore small of order $\frac1N$, and thus \eqref{EEbar} implies
\begin{equation*}
\sup_{t\in[0,T]}\lv E^{\rm part} \rv    
\leq O\lr{\frac1N} + \frac{C}{N}e^{LT},    
\end{equation*}
for any $T>0$, if $V$ satisfies \ref{H1}, or
\begin{equation*}
\sup_{t\geq 0}\lv E^{\rm part} \rv
\leq \frac{C}{N}    
\end{equation*}
for $V$ satisfying in addition Hypothesis \ref{H2}.

\section*{Acknowledgements}
{Our attention was drawn to this question by many discussions and works with Pierre Degond, whose  contribution is gratefully acknowledged. In particular, the possibility of  combining averaging and mean-field techniques was proposed in \cite{BDPZ}, a common work of two of the authors with P. Degond and D. Peurichard.} We acknowledge the generous support of the International Centre for Mathematical Sciences,  Research-in-Groups programme. P. Dobson was supported by the Maxwell Institute Graduate School in Analysis and its
Applications (MIGSAA), a Centre for Doctoral Training funded by the UK Engineering and Physical
Sciences Research Council (grant EP/L016508/01), the Scottish Funding Council, Heriot--Watt
University and the University of Edinburgh. Last but not least, the authors are very grateful to Iain Souttar, who spotted a mistake in an earlier version.

\appendix 
\numberwithin{equation}{section}

\section{Sketch of well-posedness of  the slow-fast  system \eqref{slow}-\eqref{fast}}\label{app:A}

The well-posedness for \eqref{slow}-\eqref{fast} is completely standard (under our assumptions on the potential $V$ and on the interaction force $K$), however we could not find in the literature a result that can be cited directly and that contains as a subcase systems of this form, so we just sketch the proof. 

First notice that for a general SDE driven by a Poisson process, i.e. an SDE of the form 
$$
dY_t=b(Y_{t_{-}}) dt + \sigma(Y_{t_{-}}) dN_t,  \quad Y_0=y
$$
to be well-posed, it is sufficient for $\sigma$ to be say an everywhere defined function, and $b$ to be such that the solution of the ODE $dY_t = b(Y_t) dt$ exists globally (and, in particular, starting from any initial datum).  Indeed,  the solution $Y_t$ can be iteratively constructed as follows: i) consider an outcome of Poisson events $T_1, T_2, \dots$. In between events, i.e. on each interval $[T_j, T_{j+1})$,  the dynamics follows the ODE   $dY_t = b(Y_t) dt$, so in particular we know the value of $Y_{T_{{j+1}_-}}$. The value after the jump is  given by $Y_{T_{j+1}}=Y_{T_{j+1}-}+ \sigma (Y_{T_{{j+1}_-}})$. After the jump the solution will resume evolving according to the ODE with initial datum $Y_{T_{j+1}}$. 

The strong solution (i.e. pathwise solution) to  \eqref{slow}-\eqref{fast} can be built using the same logic; to simplify notation, suppose first $N=2$ so that there are only two particles and one link, $A(t)\in \{0,1\}$.     In between firing events of either $N^{d,\ep}$
or $N^{f,\ep}$ the dynamics \eqref{slow} evolves according to either the simple Langevin equation
$$
\left\{
\begin{array}{c}
dX^{1,2, \ep}(t) = - \nabla V(X^{1,2, \ep}(t))
+ \sqrt{2D} dB^1_t\\
dX^{2,2, \ep}(t) = - \nabla V(X^{2,2, \ep}(t))
+ \sqrt{2D} dB^2_t
\end{array}
\right.
$$
or 
$$
\left\{
\begin{array}{c}
dX^{1,2, \ep}(t) = - \nabla V(X^{1,2, \ep}(t))
+ K(X^{2,2, \ep}(t) - X^{1,2, \ep}(t))+ \sqrt{2D} dB^1_t\\
dX^{2,2, \ep}(t) = - \nabla V(X^{2,2, \ep}(t))
+ K(X^{1,2, \ep}(t) - X^{2,2, \ep}(t))
+ \sqrt{2D} dB^2_t
\end{array}
\right.
$$
Under our assumptions both $\nabla V$ and $K$ are globally Lipshitz, so global well-posedness for both of the above SDEs is completely standard, once the initial datum has finite second moment.   
Similarly for any $N>2$. 
Note also that, for every $N$ fixed, there is almost surely a strictly positive amount of time between any two firing events of $N^{d,\ep}$ and $N^{f,\ep}$ (non-explosivity): with the notation used in the introduction, one can consider first the sum of $N^{d,\ep}$ and  $\breve N^{f,\ep}$,  which is still a Poisson process, with rate $\nu_d+\nu_f$. Hence, by elementary properties of simple Poisson processes, there is almost surely a strictly positive amount of time between any two firing events of $N^{d,\ep}$ and $\breve N^{f,\ep}$. But, by construction, the time between firing events of $N^{d,\ep}$ and $N^{f,\ep}$ is longer than the time between firing events of $N^{d,\ep}$ and $\breve N^{f,\ep}$, as the realization of the sum of  $N^{d,\ep}$ and $N^{f,\ep}$ is a subsequence of the realization of the sum of $N^{d,\ep}$ and $\breve N^{f,\ep}$. Observe moreover that the paths of  \eqref{slow} will be a.s. continuous,  as  the coefficients $A_{ij}^{N, \ep} \in \{0,1\}$ only have the effect of switching on or off the interaction kernel $K$.

\end{document}